\documentclass[3p]{elsarticle}

\usepackage{graphicx} % Required for inserting images
\usepackage{amsmath}
\usepackage{amssymb}
\usepackage{xcolor}
\usepackage{comment}
\usepackage{float}
\usepackage{adforn}
\usepackage{dirtytalk}
\usepackage{enumitem}
\setlist[itemize]{noitemsep}
\usepackage{algorithm}
\usepackage{algpseudocode}
\usepackage{caption,subcaption}
\usepackage{amsthm,thmtools}
\usepackage{bm, bbm}
\usepackage{mathtools}
\usepackage{subcaption}
\usepackage{hyperref}
\usepackage[capitalise]{cleveref}
\usepackage{wrapfig}

\newtheorem{theorem}{Theorem}

\declaretheorem{remark}
\title{Improving Nonpreemptive Multiserver Job Scheduling with Quickswap}
\author[1]{Zhongrui Chen}
\ead{jcpwfloi@cs.unc.edu}

\author[2]{Adityo Anggraito}
\ead{adityo.anggraito@unive.it}

\author[2]{Diletta Olliaro}
\ead{diletta.olliaro@unive.it}

\author[2]{Andrea Marin}
\ead{marin@unive.it}

\author[3]{Marco Ajmone Marsan}
\ead{marco.ajmone@imdea.org}

\author[1]{\\{Benjamin Berg}}
\ead{ben@cs.unc.edu}

\author[4]{Isaac Grosof}
\ead{izzy.grosof@northwestern.edu}

\affiliation[1]{organization={University of North Carolina at Chapel Hill},
city={Chapel Hill, NC},
country={USA}}

\affiliation[2]{organization={Università Ca’ Foscari Venezia},
city={Venice},
country={Italy}}

\affiliation[3]{organization={IMDEA Networks Institute},
city={Leganes},
country={Spain}}

\affiliation[4]{organization={Northwestern University},
city={Evanston, IL},
country={USA}}

\definecolor{unc}{RGB}{75,156,211}

\allowdisplaybreaks
\everymath{\displaystyle} 

\newcommand{\wt}{\widetilde}
\newcommand{\wh}{\widehat}

\newcommand{\E}{\mathbb{E}}
\DeclarePairedDelimiter{\norm}{\lVert}{\rVert}
\DeclarePairedDelimiter{\set}{\{}{\}}

\newenvironment{DIFnomarkup}{}{}
\newcommand{\phase}{H}
\begin{document}

\begin{abstract}
%\zhongrui{Thanks Diletta and Marco for the draft on the abstract. I have merged everything into the current version.}
Modern data center workloads are composed of \textit{multiserver jobs}, computational jobs that require multiple servers in order to run.
A data center can run many multiserver jobs in parallel, as long as it has sufficient resources to meet their individual demands.
Multiserver jobs are generally \emph{stateful}, meaning that job preemptions incur significant overhead from saving and reloading the state associated with running jobs.
Hence, most systems try to avoid these costly job preemptions altogether.
Given these constraints, a \emph{scheduling policy} must determine what set of jobs to run in parallel at each moment in time to minimize the mean response time across a stream of arriving jobs.
Unfortunately, simple non-preemptive policies such as First-Come First-Served (FCFS) may leave many servers idle, resulting in high mean response times or even system instability.
Our goal is to design and analyze non-preemptive scheduling policies for multiserver jobs that maintain high system utilization to achieve low mean response time.
\begin{comment}
Because each job may request a different number of cores, it is not always possible to select sets of multiserver jobs that fully utilize the available server resources.
Simple policies like First-Come First-Served (FCFS) may leave many cores idle, resulting in high mean response times or even system instability.
While it is possible to improve system utilization by using more complex policies that pack jobs more efficiently onto the server, these policies often require preemptions that can be prohibitively costly and impractical for data center workloads.
\end{comment}

One well-known non-preemptive scheduling policy, Most Servers First (MSF), prioritizes jobs with higher server needs and is known for achieving high resource utilization.  
However, MSF causes extreme variability in job waiting times, and can perform significantly worse than FCFS in practice.
To address this issue, we propose and analyze a class of scheduling policies called \textit{Most Servers First with Quickswap} (MSFQ) that performs well in a wide variety of cases.
MSFQ reduces the variability of job waiting times by
periodically granting priority to other jobs in the system.
We provide both stability results and an analysis of mean response time under MSFQ to prove that our policy dramatically outperforms MSF in the case where jobs either request one server or all the servers.
In more complex cases, we evaluate MSFQ in simulation.
We show that, with some additional optimization, variants of the MSFQ policy can greatly outperform MSF and FCFS on real-world multiserver job workloads.

\end{abstract}

\maketitle

\section{Introduction}
\label{sec:intro}

Modern data centers serve \emph{multiserver jobs} that occupy multiple servers simultaneously~\cite{mor-2022,tirmazi2020borg,delimitrou2014quasar,dean2008mapreduce}.
% A host server in the data center can serve multiserver jobs in parallel if it has enough dedicated CPU cores for each of these jobs.
Each multiserver job has an associated \emph{server need}, the number of servers the job requires to run, and \emph{job size}, the amount of time the job must run to be completed.
A set of multiserver jobs can run in parallel, but only if the system has enough dedicated servers for each job.
A data center \emph{scheduling policy} must select which jobs to run in parallel at every moment in time.
Given a fixed number of servers, $k$, our goal is to design a scheduling policy that minimizes the \emph{mean response time} across jobs in a stream of arriving multiserver jobs --- the average time from when a job arrives to the system until it is completed.

There are two central difficulties in designing performant scheduling policies for multiserver jobs.
First, modern datacenter workloads generally exhibit large variability in both the server needs and sizes of their jobs \cite{tirmazi2020borg}.
As a result, it is usually impossible to utilize all available servers using the set of multiserver jobs currently in the system.
In general, leaving more servers unutilized on average will lead to higher mean response time or even system instability.
Unfortunately, maximizing the number of utilized servers at a specific moment in time requires solving a knapsack problem instance.
It is even more difficult, then, to maximize the utilization of the available servers as jobs enter and exit the system over time.

Second, modern multiserver jobs are generally \emph{stateful}, meaning that job preemptions require persisting and/or reloading a significant amount of program state \cite{psychas2017non}.
As a result, job preemptions or migrations can take a prohibitively long amount of time to perform.
Due to this overhead, data centers typically employ \emph{non-preemptive} scheduling policies that avoid costly preemptions altogether.
Given these two constraints, \emph{this paper designs and analyzes new, non-preemptive scheduling policies that aim to minimize the mean response time across a stream of multiserver jobs.}

\subsection{Prior Approaches}
Much of the prior work on multiserver job scheduling uses frequent job preemptions to ensure that resource utilization remains high as jobs enter and leave the system~\cite{chen.usenix, grosof-2022-pomacs, grosof-2023-serverfilling,georgiadis-2006}.
These preemptive policies are of limited utility when processing the stateful jobs that are common in data centers.

When it comes to non-preemptive policies, there are three central approaches suggested in the literature:

\noindent\textbf{First-Come First-Served} (FCFS) is a na\"ive non-preemptive policy that serves jobs in arrival order until the system runs out of available servers.
For example, when a job with a large server need reaches the front of the queue, the system may not have enough available servers to fit this job in service.
FCFS stops scheduling additional jobs at this point, even if other jobs in the queue could fit into service.
This phenomenon, known as \emph{Head-of-the-Line blocking}, causes FCFS to underutilize servers, resulting in high mean response time.
Although FCFS is a simple policy, analyzing it has been proven difficult due to its dependence on the random arrival order of jobs with different server needs.
Only recently, \cite{grosof-2024-marcreset} derived mean response time bounds that confirm the empirical observation that FCFS performs poorly in practice.

\noindent\textbf{Most Servers First (MSF)}~\cite{bichler, beloglazov,grosof-2023-serverfilling} is a non-preemptive policy that prioritizes jobs with larger server needs.
Specifically, whenever the system has available servers, jobs are considered in descending server need order.
Jobs that find their required number of servers are put in service successively.
To understand both the benefits and the pitfalls of MSF, consider an example where jobs either need one server or $k$ servers.
We refer to this case as the \emph{one-or-all} case for multiserver jobs.
In this case, MSF serves jobs in two alternating \emph{phases}.
First, MSF serves $k$-server jobs until none remain in the system.
Then MSF serves 1-server jobs until none remain before returning to serve $k$-server jobs.
We will show in \cref{sec:throughput-optimal} that, by switching between these two phases, MSF achieves optimal long-run average resource utilization in the one-or-all case.

While one might hope that MSF leverages its high resource utilization to achieve low mean response time, we also find that MSF takes an increasingly long time to switch phases as the job arrival rate increases (see \Cref{sec:rt}).
%The issue is that, as the arrival rate increases, it will take longer for the system to complete all available jobs of a given class.
This creates a feedback loop in the system whereby many 1-server jobs accumulate while the system processes $k$-server jobs, leading to a long period of serving 1-server jobs during which many $k$-server jobs will accumulate.
The two things to note about this process are that (i) despite its name, MSF can spend long periods of time giving priority to 1-server jobs over $k$-server jobs and (ii) because the class of jobs not in service accumulates quickly, there are almost always a large number of jobs in the system under MSF.
%Specifically, even if there are many $k$-server jobs in the system, MSF may be stuck trying to drain all of the 1-server jobs in the system, which takes a long time when the arrival rate of 1-server jobs is high.
%Hence, despite the name Most Servers First, MSF can actually spend a significant fraction of time giving 1-server jobs priority over the $k$-server jobs in the system.
Figure~\ref{fig:transient} illustrates this problem via simulations that track the number of jobs in the system under MSF for the one-or-all case.
As MSF alternates between phases, jobs of the opposite class accumulate quickly in the queue.
While all jobs are eventually served, this behavior ensures that a significant fraction of arriving jobs have long queueing times, leading to a high overall mean response time.

\begin{figure}[b]
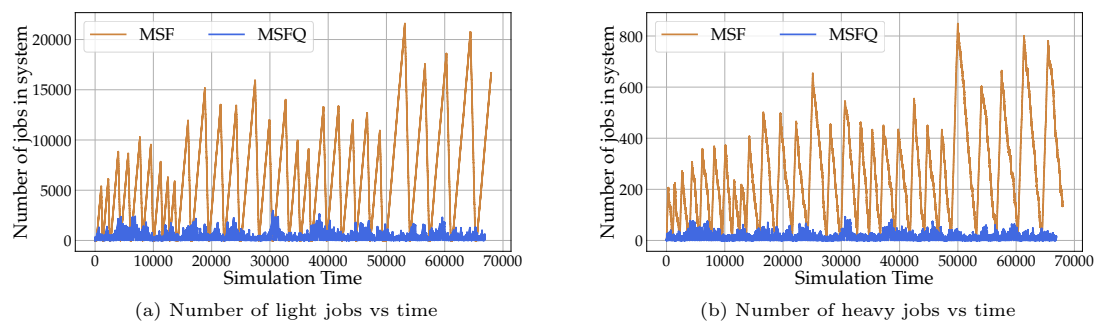

\centering
\begin{subfigure}[t]{0.45\textwidth}
    \centering
    \includegraphics[width=\textwidth]{plots/transient_small.pdf}
    \caption{Number of light jobs vs time}\label{fig:small_transient}
\end{subfigure}
\begin{subfigure}[t]{0.45\textwidth}
    \centering
    \includegraphics[width=\textwidth]{plots/transient_big.pdf}
    \caption{Number of heavy jobs vs time}\label{fig:large_transient}
\end{subfigure}\hfill
\caption{Number of jobs in the system under MSF and MSFQ where there are $k$=32 servers, 90\% of the job arrivals are 1-server jobs, mean job sizes are 1 for 1-server jobs and $k$-server jobs, and jobs arrive at a rate of 7.5 jobs/second.}
% Simulation trace of the queue occupancy in MSF and MSFQ under identical arrival streams.\newline \centering $k=32, p_1=0.9, \mu_1=\mu_k=1, \lambda=7.5$ (the system is loaded at $96\%$ of its capacity).}
\label{fig:transient}
\end{figure}
\noindent\textbf{First-Fit} is a variant of FCFS that avoids head-of-line blocking by continuing to examine jobs in arrival order after finding a job that does not fit in service.  
One might hope that this policy gets a near-optimal resource utilization without the harmful periodic behavior of MSF.
Unfortunately, in the one-or-all case, First-Fit has the same alternating behavior as MSF, but First-Fit spends even more time serving 1-server jobs.

\subsection{A New Approach: Most Servers First with Quickswap}
Our central observation is that, while MSF achieves high resource utilization, it does a poor job on switching which class of jobs are in service.
%This causes jobs of certain classes of jobs to accumulate in the system, leading to high mean response time.
Specifically, in the one-or-all case, MSF ties the decision to switch from serving 1-server jobs to serving $k$-server jobs to the time required to drain all 1-server jobs from the system.
This time, which is essentially a partial busy period in an $M/M/k$ queue, explodes as the arrival rate increases or as $k$ becomes large \cite{artalejo2001analysis}, allowing many $k$-server jobs to accumulate in the system.
To reduce mean response time, our goal is to maintain the high utilization of MSF while shortening the time the scheduling policy takes to switch between job classes.

With this motivation in mind, we propose a new class of policies called \emph{Most Servers First with Quickswap (MSFQ)}, which is designed to improve the performance of MSF in the one-or-all case.
%In order to remedy this shortcoming, in this paper, we develop and analyze an enhancement of the MSF policy that we call Most-Server-First with Quick Swap (MSFQ) to get rid of the long service phases in MSF under the one-or-all-cores setting.
Unlike MSF, which tries to drain all 1-server jobs from the system before switching, MSFQ  switches to serving $k$-server jobs whenever the number of 1-server jobs in the system falls below a threshold, $\ell$.
That is, when the number of 1-server jobs falls below $\ell$, MSFQ stops admitting jobs into service and lets all running 1-server jobs complete.
MSFQ can then begin serving $k$-server jobs.
By increasing $\ell$, we can shorten the time MSFQ requires to switch phases, reducing the mean response time compared to MSF.
Furthermore, we will prove that MSFQ maintains the high resource utilization achieved by MSF.

%Beyond proving that MSFQ is throughput optimal, like MSF, in the one-or-all condition,
Figure~\ref{fig:transient} compares the number of jobs in the system under the MSF and MSFQ policies (the yellow and blue curves, respectively). 
Setting a sufficiently high value of $\ell$ for MSFQ ($\ell=k-1$ in this case) greatly dampens the feedback that causes jobs to accumulate under MSF.
As a result, an MSFQ policy can achieve a much lower mean response time than MSF.
We will also show that MSFQ policies are much better at balancing the mean response time between different job classes.
While our MSFQ policies are tailored specifically to the one-or-all case, we will also explore generalizations of this policy to an arbitrary number of job classes.

\subsection{Contributions}
Inspired by the high resource utilization of the MSF policy, this paper formally defines and analyzes the class of MSFQ policies for scheduling multiserver jobs.
Our analysis has two main goals.
First, we aim to prove \emph{stability} results about MSFQ policies.
We say the system is stable under a scheduling policy if the policy achieves sufficiently high resource utilization such that the mean number of jobs in the system and the mean response time across jobs are both finite.
Second, we will analyze the mean response time under MSFQ policies and their variants both in theory and in simulation to show the advantages of these policies compared to prior approaches.
Specifically, the contributions of this paper are as follows.
\begin{itemize}
\item First, in Section \ref{sec:MSF}, we formally explain the shortcomings of the MSF policy in the one-or-all case.  Here, we show how excessively long periods of serving 1-server jobs cause a feedback effect that leads to poor mean response time.
\item In Section \ref{sec:MSFQ} we introduce the MSFQ policies, which uses the Quickswap mechanism to force the policies to switch phases faster.
Crucially, in Section \ref{sec:throughput-optimal} we show that any MSFQ policy matches the resource utilization of MSF by proving that MSFQ is \emph{throughput-optimal}.  Here, throughput-optimality means that MSFQ will stabilize the system whenever the system can be stabilized.
\item Then, in Section \ref{sec:rt}, we analyze the mean response time under MSFQ by approximating the Laplace Transform of the phase durations and number of jobs at the beginning of each phase.  While the MSFQ system resembles a polling system (See \cref{sec:rel-polling}), MSFQ cannot be analyzed using existing results from the polling literature.  Hence, our response time analysis of MSFQ also represents a new contribution to the extensive body of work on polling systems.
\item Finally, Section \ref{sec:sim} examines generalizations of MSFQ to real-world settings where jobs' server needs can vary widely.  We evaluate two generalizations of MSFQ, called \emph{Static Quickswap} and \emph{Adaptive Quickswap}, and evaluate these policies in simulation using traces from the Google Borg cluster scheduler \cite{tirmazi2020borg}.  Our simulations show that a well-designed Quickswap policy can improve mean response time by orders of magnitude.  Furthermore, Quickswap policies tend to achieve an equitable mean response time between the job classes as compared to a less fair priority policy like MSF.
\end{itemize}

\section{Related Work}

%Probably we want to refer also to industrial scheduler (e.g. slurm) and then refer to other contributions (e.g. serverfilling and its variations) for the sake of analytical purposes
%\ben{P1: systems like Borg and SLURM} Borg~\cite{tirmazi2020borg, borgClusterManager}, SLURM~\cite{slurm}, Hadoop Scheduler\cite{yarn}, Kube Scheduler\cite{kubernetes}
We now describe prior work on multiserver jobs from the systems and theory communities in Sections \ref{sec:rel-datacenter} and \ref{sec:rel-msj-analysis}, respectively.
We also note that the Quickswap policies analyzed in this paper bear a resemblance to prior queueing-theoretic work on polling systems.
However, because the connection between multiserver jobs and polling is somewhat indirect, we discuss the polling systems literature separately in \cref{sec:rel-polling}.

\subsection{Systems for Multiserver Job Scheduling}
\label{sec:rel-datacenter}

%\izzy{Ben: simplification pass}
Modern data centers schedule multiserver jobs across thousands of machines, supporting workloads with diverse server needs and job sizes \cite{slurm,tirmazi2020borg,yarn}.
None of these schedulers make formal performance guarantees about system stability or mean response time, generally relying on heuristics to make scheduling decisions.

SLURM~\cite{slurm} is an open-source cluster scheduler used in data centers and high-performance computing environments. It uses a combination of heuristics and a variant of FCFS scheduling called BackFilling.
While this approach can improve resource utilization by running low-priority jobs opportunistically, it requires accurate predictions of job sizes to work well, and can therefore suffer from low resource utilization in practice.
Borg~\cite{tirmazi2020borg}, Google's internal resource management system for data centers, schedules batch jobs by placing incoming jobs in an FCFS queue.  Once there is enough capacity to serve the next batch job, complex heuristics are used to assign the job to a specific set of servers.
YARN~\cite{yarn}, integrated with Hadoop, supports both FCFS and other heuristic policies aimed at optimizing fairness instead of stability or mean response time.
Hence, many systems designed to schedule multiserver jobs stand to benefit from improved scheduling policies that are accompanied by formal performance guarantees.

\subsection{Multiserver Job Scheduling in Theory}
\label{sec:rel-msj-analysis}

Prior work from the theory community on multiserver jobs has mostly focused on the stability and response time analysis of FCFS.
The stability region of FCFS was studied in \cite{rumyantsev-2017,morozov, afanaseva-2020} in the case where all job sizes follow the same exponential distribution.
Subsequently, ~\cite{grosof-2023-mama, olliaro-2023} considered the case where jobs belong to one of two job classes, deriving explicit expressions for the stability region of FCFS.
For many years, mean response time analysis of FCFS was restricted to systems with just two servers \cite{brill-1984, karatza-2007}. 
However, ~\cite{grosof-2024-marcreset} recently derived explicit bounds on mean response time that are tight up to an additive constant.
Matrix geometric approaches~\cite{anggraito-mascots-2024, anggraito-cox-2025} have also recently been used to characterize the performance of FCFS systems with two job classes under specific service time distributions.
While FCFS is becoming well-understood, all of these analyses confirm that it can perform poorly in terms of both stability and mean response time.

There is comparatively little work on more complex and efficient policies that do not require job preemptions.
For example, the well-studied MaxWeight policy is throughput optimal, but requires preemption and is computationally costly to implement in practice \cite{maguluri-2012}.
Other recent work on scheduling multiserver jobs has also been restricted to the case of preemptible jobs \cite{wcfs,grosof-2022-pomacs}.
There are two prominent examples of throughput-optimal, non-preemptive policies for scheduling multiserver jobs.
First, Randomized Timers is a throughput-optimal policy based on MaxWeight that is non-preemptive \cite{psychas2018randomized}.
Unfortunately, there is no known mean response time analysis of Randomized Timers, and the policy has been shown to perform poorly in practice.
Second, \cite{chen2025improving} recently analyzed a new class of non-preemptive policies called  Markovian Service Rate (MSR) policies.
An MSR policy precomputes a set of schedules with high resource utilization, and switches between schedules according to a continuous-time Markov chain that is independent of the system state (i.e., queue lengths).
The class of MSR policies is throughput-optimal and admits an analysis of mean response time.
However, because MSR policies do not consider queue length when switching schedules, they waste capacity unnecessarily, resulting in high mean response time.
We will show that MSFQ can significantly outperform MSR policies by considering queue length when switching schedules.

% There are approaches to mitigate straggler effects. Fork-join and DAG scheduling are hard problems that are proven to be intractable.
% References in Towards Optimality

\subsection{Polling Systems and Most Servers First}
\label{sec:rel-polling}

In the one-or-all setting, the MSF and MSFQ policies we study (see \cref{sec:MSF,sec:MSFQ}) are theoretically similar to a two-station polling system with exhaustive service, where the first station serves $k$-server jobs, and the second serves 1-server jobs.
Furthermore, the system incurs something like a switchover time when moving from 1-server jobs to $k$-server jobs.

The literature on polling systems is vast \cite{borst2018polling}.
The single-polling-station and infinite-polling-station systems are well-understood \cite{borst1997polling,foss1996polling},
and approximations for waiting time in the multiple-polling-station system
have been established \cite{borst1998waiting}.
Stability issues caused by switchover times have also been studied \cite{fricker1994monotonicity, foss1996dominance}.
However, our multiserver job system considers a mix of single-server and multiserver operations not found in the polling literature.
Furthermore, while MSF essentially uses an exhaustive service discipline for switching phases, the class of MSFQ policies uses a more generalized, threshold-based version of exhaustive service that is not analyzed in the prior work.
Hence, the analysis of MSFQ in this paper also serves as a new contribution to the literature on polling systems.

\section{Model}

%In this section, we define a general model for scheduling multiserver jobs.
% Because this general model is too complex to analyze, we simplify the general model so that response time analysis becomes possible.
% Next, we introduce a class of scheduling policies, Most Server First with Quick Swap (MSFQ), that is designed to operate within this general model under the one-or-all-cores setting in which each arriving job requests either a single server or all available servers. We then analyze MSFQ under the one-or-all-cores setting.
% Building on the strong performance of MSFQ, we propose two more policies for the general model --- Static Quickswap and Adaptive Quickswap --- both of which are shown to achieve high performance in simulation in \cref{sec:sim}.

% In the rest of the paper, we develop an analytical framework on the simplified model.

% There are approaches to mitigate straggler effects. Fork-join and DAG scheduling are hard problems that are proven to be intractable.
% References in Towards Optimality

\subsection{Multiserver Jobs}
\label{sec:MSJ}
% Our goal is to model the performance of multiserver jobs running in a data center.
We consider a system with $k$ servers.
A multiserver job can be represented by an ordered pair $(i, s)$, where $i\in\set{1,2,\cdots, k}$ is the number of servers the job needs in order to run and $s$ is the \emph{service duration} (also known as \emph{job size}), the time the job must run on the servers before completion.
Jobs occupy a fixed number of servers throughout their time in service,
and cannot be preempted: once started, a job must be run until it is complete.
We refer to this job model as the Multiserver Job (MSJ) model.

The MSJ model reflects the realities of scheduling in modern large-scale compute clusters.
Specifically, running jobs typically cannot be preempted because they are \emph{stateful} and preemption would destroy this working state \cite{psychas2017non}.
Additionally, the MSJ model does not aim to capture \emph{straggler effects} where a job's tasks on some servers finish earlier than others.
While the straggler effect is captured by more detailed models such as fork-join queueing models \cite{forkjoin}, these models are notoriously intractable to analyze.
Furthermore, modern systems employ a variety of techniques to mitigate the straggler effects \cite{ananthanarayanan2013effective}.
As a result, real-world systems such as Google Borg make scheduling decisions based on fixed server needs and largely ignore straggler effects within a job.

We consider workloads composed of different \emph{job classes}, where class-$i$ jobs all need $i$ servers.
We consider serving a stream of multiserver jobs, where   class-$i$ jobs arrive according to an independent Poisson process with rate $\lambda_i$.
We further assume the service durations of class-$i$ jobs are i.i.d. exponentially distributed random variables such that $S_i\sim \exp(\mu_i)$ for any class $i$.

We define an \emph{arrival rate vector} ${\bm \lambda}=(\lambda_1, \lambda_2, \cdots, \lambda_k)$ and a \emph{completion rate vector} ${\bm \mu}=(\mu_1, \mu_2, \cdots, \mu_k)$.
Let $\lambda$ denote the total arrival rate of multiserver jobs into the system.
Hence, $\lambda=\norm{{\bm\lambda}}_1$.
Let $p_i$ be the fraction of arriving jobs belonging to class $i$.
Equivalently, $p_i=\lambda_i/\lambda$.

We define a \emph{feasible schedule} as a multiset of classes of multiserver jobs that can run in parallel,
obeying the rule that the total number of servers requested does not exceed $k$.
We use ${\bf u}=(u_1, u_2, \cdots, u_k)$ to denote a feasible schedule in the multiserver system, where it puts $u_i$ class-$i$ jobs in service and
$\sum\nolimits_{i=1}^k i \; u_i\le k,$
as the total server demand cannot exceed $k$.

In the one-or-all setting, $\lambda_i=0$ for all $1< i<k$ because jobs can only request one server or all servers in the system.
In this case, $\lambda=\lambda_1+\lambda_k$.
A feasible schedule in this case can either be $u_k=1$ and $u_i=0$ for $i<k$, or $u_1\le k$ and $u_i=0$ for $i>1$.

A scheduling policy ${\bf u}(t)=(u_1(t), u_2(t), \cdots, u_k(t))$ picks a feasible schedule at every time $t$,
subject to the requirement that no job is preempted: the service policy at time $t$
must contain all jobs whose service has started but not completed by time $t$.
We allow the scheduling policy to select a $u_i(t)$ value that exceeds the
number of class-$i$ jobs available for some job class $i$.
In general, a scheduling policy may depend on the system state as well as policy-specific state.

We model the system state with a pair of vectors: The total \emph{occupancy vector} $\mathbf{n}(t)$,
where $n_i(t)$ is the number of class-$i$ jobs in the system at time $t$,
and the service vector $\mathbf{u}(t)$, chosen by the scheduling policy and discussed above.
In addition, we allow a general policy-specific Markovian state $z$.
The triple $({\bf n}, {\bf u}, z)$ represents the system state and forms a countably infinite Markov chain.

For notational convenience, in the one-or-all setting, we neglect all the 0's in the vectors and abbreviate $\mathbf{n}(t)$ as $\mathbf{n}(t)=(n_1(t), n_k(t))$ and $\mathbf{u}(t)$ as $(u_1(t), u_k(t))$.
Thus, we can represent the system state with a 5-tuple $(n_1, n_k, u_1, u_k, z)$.
% We are interested in the steady-state behavior of the system as $t\rightarrow \infty$. Let ${\bf Q}$ be the limiting distribution, $${\bf Q}\sim\lim_{t\rightarrow\infty} {\bf Q}(t).$$
% Let $Q$ be the number of multiserver jobs in the system.
% Then $$\E[Q]=\frac{1}{\lambda}\langle {\bm \lambda}, \E[{\bf Q}]\rangle.$$
% We can then apply Little's Law to determine the mean response time across all multiserver jobs in the system, $\E[R]$, as $$\E[R]=\frac{\E[Q]}{\lambda}.$$
% Therefore, to analyze $\E[R]$, it suffices to analyze $\E[Q]$.
% The rest of the paper will focus on the mean queue length of all job classes of the system, $\E[{\bf Q}]$.
\subsection{Stability Region}
%\zhongrui{Rewriting}
For a given policy $p$, the mean response time of the system may or may not be bounded.
Equivalently, the system may or may not be positive recurrent.
Considering the system where job size distributions and job classes are fixed and we vary the arrival rates, we define the set of arrival rates such that the mean response time is finite under policy $p$ as the \emph{policy stability region} $\mathcal{C}_p$.
Formally,
$\mathcal{C}_p=\set{\bm{\lambda} \mid \E[T_p(\bm{\lambda})]<\infty},$ where $T_p(\bm{\lambda})$ is the response time of multiserver jobs under policy $p$ when the arrival rate vector is $\bm{\lambda}$.
We also define the \emph{system stability region} $\mathcal{C}$ as the set of arrival rates such that there exists a policy that stabilizes the system.
Therefore, the system stability region is the union of the stability regions of all possible policies.
When $\mathcal{C}_p=\mathcal{C}$ for some policy $p$, we say that the policy $p$ is throughput-optimal.
In other words, a policy $p$ is throughput-optimal if it stabilizes the system whenever there exists a stable policy.
\subsection{General Notation}
\label{sec:general-notation}

We define the notation $\Sigma(X, Y)$ as the independent sum of $X$ copies of the random variable $Y$:
$\Sigma(X, Y) := \sum\nolimits_{i=1}^X Y_i,$
where $X \ge 0$ is an integer-valued random variable and the $Y_i$'s are i.i.d samples of $Y$.
We define the notation $(x)^+$ as the \emph{positive part} of $x$. Specifically, $(x)^+ := \max(x, 0)$.
We define the notation $\wh{X}(z)$ as the z-transform of the probability mass function of an integer-valued random variable $X$ and $\wt{Y}(s)$ as the Laplace-Stieltjes transform of the probability density function of a continuous random variable $Y$.

\section{Policies}
\label{sec:policies}
% \ben{maybe subsection of model?}
In this section, we define the policies we are using in the paper.
We first define the Most Servers First (MSF) policy~\cite{beloglazov, maguluri-2012},
which is known to have shortcomings as discussed in \cref{sec:intro}.
Then, based on the MSF policy, we develop the MSF Quickswap (MSFQ) policies for the setting where there are only class-$1$ and class-$k$ jobs in the system.
We analyze the mean response time of an MSFQ policy in Section~\ref{sec:rt} and show response time performance improvements with analysis and simulations in Section~\ref{sec:sim}.
Inspired by the improvements obtained with the MSFQ policies, we develop the Static Quickswap policy and the Adaptive Quickswap policy,
%and the class of Phases Quickswap policies
each of which is a generalization of MSFQ to support arbitrary sets of job classes.
We show that the mean response time performance of Static Quickswap and Adaptive Quickswap each compares favorably to MSF in simulations based on real datacenter traces in Section~\ref{sec:sim}.

\subsection{Most Servers First}
\label{sec:MSF}
We analyze the Most Servers First (MSF) policy as described in~\cite{beloglazov, maguluri-2012}.
Specifically, we define MSF as a non-preemptive policy that favors jobs with the highest server demands.
Whenever a job arrives or completes, MSF tries to put as many additional jobs as possible into service, starting with the job that demands the most servers and moving in descending order of server demands.
This process ends when either all servers are utilized or MSF has considered all jobs in the queue.

In the one-or-all case where jobs either require 1 server or $k$ servers, MSF has a somewhat simpler structure.
At any moment in time, the policy either serves 1 class-$k$ job or up to $k$ class-$1$ jobs.
The two job classes are never served simultaneously.
We describe this structure by saying that MSF undergoes two \emph{phases}.
In phase 1, MSF serves exclusively class-$k$ jobs one at a time. Doing so, MSF uses all servers in the system and is thus very efficient.
In phase 2, MSF serves exclusively class-$1$ jobs.
During phase 2, there can be up to $k$ class-$1$ jobs in service, depending on how many class-$1$ jobs are in the system.
Whenever the number of class-1 jobs in service is less than $k$, some of the servers' service capacity is wasted.
MSF only switches between phases when it runs out of jobs of the current class.

Considering the phases of MSF in the two-class case demonstrates why this policy can lead to poor performance, particularly as load becomes high.
The time to complete phase 1 looks like the busy period of an $M/M/1$ system started by the jobs that arrived in the prior phase 2.
The time to complete phase 2 looks like the busy period of an $M/M/k$ system started by the jobs that arrived during the prior phase 1.
As load increases, both phases will become longer, causing more class-$1$ jobs to arrive during phase 1 and vice versa.
In this way, MSF amplifies the effect of increased load on mean response time.

%MSF continuously admits jobs from the queue with the greatest server need that fits in the server.

\subsection{Most Servers First Quickswap (MSFQ)}
\label{sec:MSFQ}

% \zhongrui{define H and N's here}

To reduce the load-amplifying effect found in MSF, we introduce the Most Server First Quickswap (MSFQ) policies in the one-or-all setting.
% \ben{name this in the model or something}
The goal of MSFQ is to shorten the duration of phase 2 of the MSF policy so that fewer class-$k$ jobs are allowed to build up during this phase.
Specifically, an MSFQ policy is associated with a threshold, $\ell$, that is used to shorten the periods where class-$1$ jobs are in service.
When the number of class-1 jobs drops below $\ell$, the system stops allowing class-1 arrivals to enter service.

We define the MSFQ policy formally via the following phases with a threshold $\ell \in [0, k-1]$:
\begin{itemize}
    \item Phase 1: Serve class-$k$ jobs exclusively until none remain ($n_k=0$).
    \item Phase 2: Serve class-1 jobs until there are less than $k$ class-1 jobs in the system ($n_1<k$).
    \item Phase 3: Serve class-1 jobs until there are at most $\ell$ class-1 jobs in service ($n_1\le \ell$).
    \item Phase 4: Complete the class-1 jobs that are already in service ($n_1=0$ at the end of the phase).  New class-1 arrivals are not allowed to enter service during this phase.
\end{itemize}
At the end of phase 4, the server returns to phase 1. 
Note that phase 2 and phase 3 are similar, but it will be convenient to treat them separately in our analysis.

By analyzing response time of the MSFQ system in \Cref{sec:rt}, we find that the mean response time is dependent on the number of jobs in the system at the beginning of the phase, and the \emph{phase duration} --- the amount of time from when the system enters phase $i$ until phase $i$ completes.
We let random variable $N_i^L$ denote the number of light (class-1) jobs and $N_i^H$ denote the number of heavy (class-$k$) jobs at the beginning of phase $i$.
We also let the random variable $\phase_i$ denote the duration of the $i$th phase.
% Each phase duration depends on the number jobs that are in the system at the beginning of the phase.

Note that, when $\ell=0$, our MSFQ policy is the same as the MSF policy.
We show that the class of MSFQ policies, irrespective of the threshold $\ell$, is throughput optimal in Section~\ref{sec:throughput-optimal}.
The intuition behind throughput-optimality follows from the fact that we always fully utilize all the servers outside of the switchover period (phases 3 and 4) when there are enough jobs waiting in the system.
% We analyze the mean response time performance of our MSFQ policy under the one-or-all-cores setting in Section~\ref{sec:rt}.
By searching over the choices of $\ell$, we show that our MSFQ policy offers significant performance improvement over MSF.
This confirms that the load-amplifying effect found in MSF has a huge impact on response time performance.

%An MSFQ policy selects an $\ell$ as threshold and runs the four phases cyclically.
\subsection{Static Quickswap}
Our one-or-all MSFQ analysis relies heavily on the fact that there are only two classes of jobs in the system.
The key difficulty in generalizing the MSFQ policies to arbitrary sets of job classes lies in selecting the next job to serve, once the phase corresponding to the current class of jobs is complete.
% Our analytical framework does not scale trivially with the number of classes because the interdependence from the load-amplifying effect gets complicated as the number of classes grows.
% Even if we cannot easily analyze the response time performance as the number of classes grows, we conjecture that our quickswap method is effective for more general settings and we validate that in simulations and traces.

We therefore define the \emph{Static Quickswap} scheduling policy,
which cycles through all classes of jobs in a fixed order, with the following two phases for each class of jobs, $i$:
\begin{itemize}
    \item Working Phase: In class $i$'s working phase, 
    we serve class-$i$ jobs exclusively until the number of idle servers exceeds $k-\ell$.
    Formally, in class $i$'s working phase we set $u_i=\lfloor k/i \rfloor$, and we set $u_j=0$ for all $j\neq i$.
    % $u_j=\begin{cases}
    %     \lfloor N/i \rfloor & j = i\\
    %     0 & j\neq i
    % \end{cases}$
    \item Draining Phase: During class $i$'s draining phase, we complete the class-$i$ jobs that are still in service at the end of the working phase. New class-$i$ arrivals are not allowed to enter service during this phase.
\end{itemize}
When a given class's draining phase is complete, the next class in the cycle is served.
We do not focus on the choice of cyclic ordering of the phases -- we leave studying the effects of that ordering to future work.

% As a first step moving from one-or-all-cores setting to more general settings, we consider the case where there are class-$i$ jobs only if $i$ divides $N$, which we call the \emph{divisible setting}.
% % , we develop a policy called Divisible MSFQ.
% We develop a policy called Divisible MSFQ for the system under divisible setting, which is defined by the following phases:
% \begin{itemize}
%     \item Working Phase: Run schedule $u_j=\begin{cases}
%         \lfloor N/i \rfloor & j = i\\
%         0 & j\neq i
%     \end{cases}$ and serve class-$i$ jobs exclusively until the number of idle servers exceeds $N-\ell$.
%     \item Draining Phase: Complete the class-$i$ jobs that are already in service. New class-$i$ arrivals are not allowed to enter service during this phase.
% \end{itemize}

% Our Divisble MSFQ policy serves all job classes in some order in a working phase, following by a draining phase in between working phases. \textcolor{green}{Should we say something about the ordering of job classes?}
% \zhongrui{We did not optimize the choice of order here. Should we say in a pre-determined order?}
% The intuition behind throughput-optimality of MSFQ extends to this case, as we always fully utilize all the cores outside of the switchover period (draining phase) when there are enough jobs waiting in the system.
% While we do not have formal proofs on throughput-optimality, we conjecture that this holds and defer it to future works.
In \cref{rem:general-stability},
we give a proof sketch that the Static Quickswap policy achieves optimal stability region whenever all classes of jobs perfectly divide $k$, and thus fully utilize the $k$ servers.

We will empirically show that the response time performance of Static Quickswap compares favorably to MSF in Section~\ref{sec:sim}.

\subsection{Adaptive Quickswap}
Inspired by MSFQ and Static Quickswap, we develop a policy called Adaptive Quickswap,
which allows multiple classes of jobs at the same time.
In the general setting, if the number of servers required by a class does not perfectly divide the total number of servers,
the system cannot fully utilize the server when serving some class of jobs exclusively.
Hence, our Adaptive Quickswap policy prioritizes jobs to serve in MSF order
and switches when it finds that serving these jobs becomes inefficient. 
The Adaptive Quickswap policy first admits jobs according to MSF order and then operates according to the following phases:
\begin{itemize}
    \item Working Phase: Whenever servers become available, the job in the queue with the largest server need that is at most the number of unoccupied servers is admitted to service.
    This continues until the quickswap is triggered.
    \item Quickswap trigger: Switch from the working phase to the draining phase
    when there is a job class that is in the queue and not in service,
    and every job class in service has no jobs of that class waiting to receive service.
    \item Draining Phase: No jobs may enter service,
    except for the job in the queue with the largest server need.
    Once this job has entered service, switch to the working phase.
\end{itemize}

% \textcolor{red}{The above is a description of the transitions between phases, more than of what happens during phases.}

% \subsection{Phases Quickswap}
% \zhongrui{Placeholder. Complete if we have enough results here.}
% We develop a class of policies called Phases Quickswap policies.
% This class of policies preselects $K$ feasible schedules for $K$ phases and runs these phases cyclically.
% Because the jobs are non-preemptible, a Phases Quickswap policy cannot directly switch from one phase to another.
% Therefore, the system enters a \emph{switchover period} between phases.
% In the switchover period, class-$i$ jobs can enter service only if the number of class-$i$ jobs in the system is less than the number of class-$i$ jobs specified in the feasible schedule in the next phase.
% The switchover period ends when all running jobs belong to the feasible schedule in the next phase.
% Not sure about the section title, but maybe we want a section presenting all the MSF variations we are studying in this work.
\section{Analysis of the MSFQ Policy}
\label{sec:results}
% \zhongrui{Make it clear that the approximation only happens in 5.3 but not before that.}

In this section, we analyze the Most Servers First Quickswap (MSFQ) policies under the one-or-all setting, where there are only class-1 and class-$k$ jobs.
We call the class-1 jobs light jobs and the class-$k$ jobs heavy jobs.
%, and the case where the jobs' server needs divides the total number of servers, $k$.
We begin by proving that any MSFQ policy is throughput-optimal in the simplified setting.

\begin{restatable}[MSFQ Throughput-Optimality]{theorem}{msfqstability}
        \label{thm:opt-stability}
        In the one-or-all MSJ system,
        Most Servers First with Quickswap (MSFQ) policies
        have optimal stability region
        for all thresholds $\ell$, $0\le \ell < k$.
    \end{restatable}

We then analyze the mean response time under MSFQ.
We observe that the mean response time depends on two central factors: the amount of time the policy spends in each phase, and how many jobs are in the system at the beginning of each phase.
We therefore approximate the Laplace-Stieltjes transforms of the distributions of phase durations and the z-transforms of the number of jobs distributions at the beginning of each phase.
We then show how to use these transforms to approximate the mean response time for light and heavy jobs.
These three steps give the following summary theorem regarding $\E[T]$.

\begin{restatable}[MSFQ Summary Theorem]{theorem}{msfqsummary}
\label{thm:summary}
The mean response time under MSFQ, $\E[T]$, depends on the first and second moments of $\phase_i$ and $N_i$ for all phases $i$.
Hence, one can compute $\E[T]$ using the transforms $\wh{N_i^L}(z)$ and $\wh{N_i^H}(z)$ for all $i$ using Lemmas~\ref{lem:fraction}-\ref{thm:absorption}.
\end{restatable}

All of our analysis applies to the original MSF policy,
by setting the Quickswap parameter $\ell$ to 0.

\subsection{Throughput-optimality of MSFQ}
\label{sec:throughput-optimal}
    We now prove that MSFQ policies achieve optimal stability region,
    for any Quickswap threshold parameter $\ell$.
    We start by lower-bounding its stability region:
    % \begin{theorem}
    %     \label{thm:pos-stability}
    %     In the 2-class one-or-all MSJ system,
    %     the Most Servers First with Quickswap policy
    %     is positive recurrent for all thresholds $\ell$, $0\le \ell < k$, whenever
    %     $
    %         \frac{\lambda_1}{k \mu_1} + \frac{\lambda_k}{\mu_k} < 1.
    %     $
    % \end{theorem}
    \begin{restatable}{theorem}{stability}
        \label{thm:pos-stability}
        In the one-or-all MSJ system, the Most Servers First with Quickswap policies are positive recurrent for all thresholds $\ell$, $0\le \ell < k$, whenever
        $
            \frac{\lambda_1}{k \mu_1} + \frac{\lambda_k}{\mu_k} < 1.
        $
    \end{restatable}
    \begin{proof}
        We prove this theorem via the Foster-Lyapunov theorem with a carefully designed Lyapunov function.  See \DIFnomarkup{\ref{app:stability}} for a full proof.
     \end{proof}

    We now upper bound the optimal possible stability region in the one-or-all system:

    \begin{theorem}
        \label{thm:neg-stability}
        In the one-or-all MSJ system,
        no scheduling policy is stable if
        $
            \lambda_1/k \mu_1 + \lambda_k/\mu_k \ge 1,
        $
    \end{theorem}
    \begin{proof}
        Define a job's \emph{work} as the product of server need and mean service duration,
        scaled down by $k$.
        The system can never complete work at rate above 1, because there are $k$ servers.
        Work arrives to the one-or-all system at a rate of $\lambda_1/k \mu_1 + \lambda_k/\mu_k$.
        If this rate is 1 or more, the system cannot be stable.
        This argument can be further formalized by comparing with a resource-pooled M/G/1.
    \end{proof}

    The throughput-optimality of the MSFQ policy now follows
    by combining \cref{thm:pos-stability,thm:neg-stability}:
    \msfqstability*

     \begin{remark}
        \label{rem:general-stability}
     Using similar arguments to Theorems \ref{thm:pos-stability} and \ref{thm:neg-stability}, we can also bound the stability region of the Static Quickswap policy with arbitrary job classes.
     The main difference in this case is that the number of servers required by class-$j$ jobs may not perfectly divide the total number of servers, $k$, leading to wasted capacity.
     Hence, in this general case, the sufficient condition for stability becomes
     $
         \sum\nolimits_j \frac{\lambda_j}{\lfloor k/j \rfloor \mu_j} < 1,
     $
     where the floor function accounts for the wasted capacity when serving class-$j$ jobs.
     The necessary condition for stability from \Cref{thm:neg-stability} remains unchanged:
     $
         \sum\nolimits_j \frac{\lambda_j}{(k/j) \mu_j} < 1.
     $
      Hence, unless $j$ divides $k$ for all classes, the Static Quickswap policy is not throughput-optimal in this general case.
     \end{remark}

\subsection{Approximations in Response Time Analysis}
\label{sec:rt:approx}

To make the MSFQ system more tractable to response time analysis, we assume in \cref{sec:rt,sec:cycle-analysis} that there is at least 1 heavy job in the system at the beginning of phase 1, and at least $k$ light jobs in the system at the beginning of phase 2.
This approximation ensures that all the phases are not skipped in a cycle, thus making the system easier to analyze despite being highly accurate as shown in \Cref{sec:sim}.
Intuitively, when the system load gets high, there would be at least 1 heavy job in the system at the beginning of phase 1 and at least $k$ light jobs in the system at the beginning of phase 2 with high probability.

\subsection{Response time analysis}
\label{sec:rt}
Given that the MSFQ policies are throughput-optimal, we would like to analyze the mean response time of a stable MSFQ system.
% We let $R$ denote the response time of a job in the system.
To analyze $\E[T]$, we analyze the conditional response time of a light job or a heavy job given the phase in which it arrives.

Specifically, in \Cref{thm:exceptional-first} we will analyze the conditional mean response time, $\E[T^H_1]$,
of a heavy job that arrives in phase 1,
and in \Cref{thm:excess} the conditional mean response time, $\E[T_{2,3,4}^H]$, of a heavy job that arrives in any of phases 2, 3, or 4.
Similarly, we analyze the conditional mean response times $\E[T^L_{1,4}], \E[T^L_2],$ and $\E[T^L_3]$ of light jobs that arrive in either phases 1 or 4, phase 2, and phase 3, respectively.
% Specifically, let $T_i^L$ denote the response time of a light job that arrives during phase $i$ and let $T_i^H$ denote the response time of a heavy job that arrives during phase $i$.
We also analyze $m_i$, the fraction of time the system spends in phase $i$.
% We will analyze each $\E[T_i^L]$, $\E[T_i^H]$ and $m_i$ separately, and then compute $\E[T]$ by conditioning. %in Theorem~\ref{thm:summary}.

We then characterize the mean response time under MSFQ, $\E[T]$, as follows:
%$$\E[T]=\E[T\mid \textrm{heavy job}]\frac{\lambda_k}{\lambda} + \E[T\mid \textrm{light job}]\frac{\lambda_1}{\lambda}.$$
\begin{align}
    \E[T] =\frac{\lambda_k}{\lambda}\left(\E[T_1^H]m_1+\E[T_{2,3,4}^H](m_2+m_3+m_4)\right)+\frac{\lambda_1}{\lambda}\left(\E[T_{1,4}^L](m_1+m_4)+\E[T_2^L]m_2+\E[T_3^L]m_3\right).\label{eq:er}
\end{align}
% We can further condition each term in this expression based on which phase the job arrives in to obtain
% \begin{align*}
%         \E[T\mid \textrm{heavy job}] &= \sum_{i=1}^4 m_i \E[T_i^H] \qquad\mbox{and}\qquad \E[T\mid \textrm{light job}] = \sum_{i=1}^4 m_i \E[T_i^L].
%     \end{align*}
% Hence, we can write $\E[T]$ as    
%     \begin{align}
%         \E[T] &= \frac{\lambda_k}{\lambda} \sum_{i=1}^4 m_i \E[T_i^H] + \frac{\lambda_1}{\lambda} \sum_{i=1}^4 m_i \E[T_i^L]\label{eq:er}.
%     \end{align}
% \ben{these defs hopefully can be moved earlier}
% Next, we define the rest of the terms we need in order to analyze $\E[T_i]$ in \eqref{eq:er}.
% Let the random variable $\phase_i$ denote the \emph{phase duration} of the $i$th phase --- the time amount of time from when the system enters phase $i$ until phase $i$ completes.
% Each phase duration depends on the number jobs that are in the system at the beginning of the phase.
% Let $N_i^L$ denote the number of light (class-1) jobs and let $N_i^H$ as the number of heavy (class-$N$) jobs at the beginning of phase $i$.

Our analysis of $\E[T]$ has three steps.  
First, in \cref{lem:fraction}, we show that each $m_i$ can be computed as a function of $\E[H_i]$, the mean duration of phase $i$.
Second, in \cref{thm:exceptional-first,thm:excess,thm:absorption}, we derive explicit expressions for the conditional mean response times listed above,
and show that these expressions depend on just the first and second moments of $\phase_i$ and $N_i$ for all $i$.
%Obviously, the $m_i$ terms in \eqref{eq:m} also depend only on the first moments of $\phase_i$ for all $i$.
Hence, to compute $\E[T]$ it suffices to find the transforms, $\wh{N_i^L}(z)$, $\wh{N_i^H}(z)$, and $\wt{\phase_i}(s)$ for all $i$.
Third, we compute the necessary transforms in Section~\ref{sec:cycle-analysis}.

% The three steps of our argument will be proven in Lemmas \ref{thm:exceptional-first}-\ref{thm:absorption}.  We first prove these lemmas before finishing the section by proving Theorem \ref{thm:summary}.

We begin by showing in \cref{lem:fraction} that the fraction of time $m_i$ that the system spends in state $i$
is proportional to the phase length $\E[H_i]$.
% If the phase cycles of the MSFQ policy formed a renewal process, one could apply a renewal-reward argument to claim that
% \begin{align}
%         m_i&=\frac{\E[\phase_i]}{\sum_{i=1}^4 \E[\phase_i]}\label{eq:m}.
% \end{align}
% If the phase cycles of the MSFQ policy formed a renewal process, one could apply a renewal reward argument \cite{harchol2013performance} to prove \cref{lem:fraction}.
% Unfortunately, the phase cycles do not form a renewal process.
% A long cycle accumulates many arrivals, leading to another long cycle.
% Hence, to relate $m_i$ and $\E[H_i]$, we first prove a generalization of the renewal-reward theorem, Theorem \ref{thm:extended-rr}, and then use this theorem to prove \cref{lem:fraction}.

%    Fortunately, we can use the positive recurrence of the overall system to prove a generalized renewal-reward theorem that applies even in this situation. We apply \cref{thm:extended-rr}, where the ``distinguished event'' is the beginning of a new cycle of phases.
    % we note that the total cycle length is independent if we condition on $N_1^H$, the number of light jobs the system has at the beginning of phase 1.
    
    % Hence, we can extend the renewal reward theorem in Theorem~\ref{thm:extended-rr} to show that the theorem holds when the renewal cycles are only conditionally independent --- they are independent when some external states are equal.
    \newcommand{\reward}{R}
    % \begin{restatable}[Extended renewal reward theorem]{theorem}{renewaltheory}
    %     \label{thm:extended-rr}
    %     Consider a positive continuous-time stochastic process $\{I_t\}_{t \ge 0}$.
    %     Associated to the stochastic process is a reward function $R(t)$ specifying the total reward up to time $t$,
    %     and certain distinguished events which occur at times $\{K_n\}_{n \ge 1}$, whenever the process $\{I_t\}_{t \ge 0}$ enters a distinguished subset of its state space.
        
    %     We define a discrete-time stochastic process $\set{(X_n, Y_n, R_n), n\ge 1}$,
    %     where $X_n$ is the interarrival time $K_{n+1} - K_n$ of the $n$th step,
    %     where $Y_n = I_{K_n}$ is the state at the beginning of step $n$,
    %     and $R_n = R(K_{n+1}) - R(K_n)$ is the reward of step $n$.
    %     We assume the mean duration of each step and the mean reward per step are finite: $\E[X]<\infty, E[R]<\infty$.
    %     We assume that the step durations and step rewards are identically and independently distributed given the state:
    %     \begin{align*}
    %         [X_i \mid Y_i = y] \sim [X_j \mid Y_j = y]\, \forall i,j,y, 
    %         \qquad 
    %         [R_i \mid Y_i = y] \sim [R_j \mid Y_j = y]\, \forall i,j,y.
    %     \end{align*}
        
    %     Then with probability 1, the average reward accrued by time $t$ in the underlying continuous-time process converges almost surely
    %     to the ratio of the mean reward per step and the mean duration of a step: 
    %     \begin{align*}
    %         \frac{R(t)}{t} \rightarrow \frac{\E[R]}{\E[X]} ~as~ t\to\infty.
    %     \end{align*}
    % \end{restatable}
    
% \begin{proof}
%     See~\ref{app:extended-rr}.
% \end{proof}

\begin{restatable}{lemma}{fractionm}
\label{lem:fraction}
The fraction of time the MSFQ system spends in phase $i$ is
% \begin{align*}
    $m_i=\frac{\E[\phase_i]}{\sum_{i=1}^4 \E[\phase_i]}$.
% \end{align*}
\end{restatable}
\begin{proof}
Follows directly from Palm inversion \cite{leboudec2010performance}.
\end{proof}

Next, we will show that the conditional mean response times in \eqref{eq:er} depend only on the first and second moments of $H_i$ and $N_i$ for each phase.
% This is obvious for the $\E[H_i]$ terms.
We handle $\E[T_1^H]$ and $\E[T_2^L]$ in \Cref{thm:exceptional-first}, $\E[T_{2,3,4}^H]$ and $\E[T_{1,4}^L]$ in \Cref{thm:excess}, and $\E[T_3^L]$ in \Cref{thm:absorption}.

%We now prove Lemmas~\ref{thm:exceptional-first}-\ref{thm:absorption} before returning to prove Theorem~\ref{thm:summary} at the end of this section.

To analyze $\E[T_1^H]$ and $\E[T_2^L]$, we relate these terms to an \emph{$M/G/1$ with Exceptional First Service} (EFS) system~\cite{exceptional-first} as defined in \Cref{remark:ex-first}.

\begin{restatable}[EFS system]{remark}{exceptionalfirst}
    \label{remark:ex-first}
    As stated in~\cite{exceptional-first}, an \emph{$M/G/1$ with Exceptional First Service} (EFS) system serves two different classes of jobs.
    Normally, job sizes are drawn i.i.d. according to some job size distribution $S$.
    However, the first job in each busy period experiences exceptional first service, and has a job size distributed as $S'$.
    Let $\E[W^{EFS}(\lambda, S, S')]$ be the mean work in an EFS system with arrival rate $\lambda$.  From~\cite{exceptional-first}, we have
    \begin{align*}
        E[W^{EFS}(\lambda, S, S')] &= \frac{\lambda E[S^2]}{2(1-\lambda\E[S])} + \frac{\lambda (E[S'^2] - E[S^2])}{2(1-\lambda\E[S] + \lambda E[S'])}.
    \end{align*}
    Let $p^{EFS}(\lambda, S, S')$ be the probability that a job arrives to an empty system and experiences exceptional service.
    We have that
    \begin{align*}
        p^{EFS}(\lambda, S, S') = \frac{1 - \lambda \E[S]}{1 - \lambda\E[S] + \lambda E[S']}.
    \end{align*}
\end{restatable}

%We will now relate $\E[T_1^H]$ and $\E[T_2^L]$ to an EFS system in \Cref{thm:exceptional-first}.
%\end{restatable}
%\begin{proof}

%    \zhongrui{Do we need a renewal reward argument for $m_i$'s here?}
%    \ben{Is this a renewal process?}
%    \zhongrui{Yes! The system runs the phases cyclically. If we set up an indicator function for whether the system is in phase 1, it gives us $m_1=\frac{\E[\phase_1]}{\E[\textrm{cycle length}]}$}
%    \ben{yeah, but cycles aren't independent}
%    \zhongrui{Sorry. I think we can work with the definition of $\phase_1$. $\phase_1$ is exactly the definition of time spent in phase 1 per cycle. We only need to refer to that for $m_1$. Then, the cycle length would just be sum of all $T's$ due to linearity of expectations. Wait we can't divide expectations here.}
%    \ben{ok, think this over, we'll discuss at 2}
%    
%    \zhongrui{Draft: $\phase_1$, $\phase_2$ are only dependent through $N$'s. $\phase_4$ and $\phase_1$ are positively correlated. $\phase_{3}$ is independent from the rest. What we want: long run average fraction of time spent in each phase.}
%
%    \zhongrui{We might need Alternating Renewal Process with 4 states. That is, we need a Markov Renewal Process / Semi Markov Process.}
    % $m_i$ holds if we cite this: https://www.sciencedirect.com/science/article/abs/pii/B9780124077959000062
% \end{proof}
% \ben{write as function of $\E[W_q]$}
\begin{restatable}{lemma}{r1lr2s}
\label{thm:exceptional-first}
    The mean response time of heavy jobs arriving into phase 1, $\E[T^H_1]$, and the mean response time of light jobs arriving into phase 2, $\E[T^L_2]$, can be characterized as follows:
    \begin{align*}
        \begin{cases}
        \E[T^H_1] &= \frac{\E[W^{EFS}(\lambda_k, S_k, \sum(N_1^H,  S_k))]}{1-p^{EFS}(\lambda_k, S_k, \sum(N_1^H,  S_k))} + \frac{1}{\mu_k} \vspace{0.5em} \\
        E[T_2^L] &= \frac{\E[W^{EFS}(\lambda_1, S_1/k, \sum(N_2^L-k+1, S_1/k))]}{1-p^{EFS}(\lambda_1, S_1/k, \sum(N_2^L-k+1, S_1/k)))} + \frac{1}{\mu_1}
        \end{cases},
    \end{align*}
    where $\Sigma(X,  Y)$ is defined as $\sum\nolimits_{i=1}^X Y_i$.
    These expressions only depend on the first and second moments of $S_1$, $S_k$, $N_1^H$, and $N_2^L$.
\end{restatable}
\begin{proof}
    % An $M/G/1$ system with exceptional first service is defined as an $M/G/1$ system but any job that arrives into an empty system will experience a different service duration distribution.
    % We call a job that arrives into an empty system a \emph{job with exceptional first service}.
    We compare a tagged heavy job that arrives into phase 1 with a job in the EFS system to compute its response time.
    Specifically, we compare the mean work in the system a tagged heavy job sees on arrival during phase 1 to the mean work a tagged job that does not receive exceptional first service sees.
    
    Consider the case where jobs arrive into an EFS system with rate $\lambda_k$, job sizes are i.i.d. exponentially distributed according to $S_k$ except for the jobs that receive exceptional first service, whose job sizes are sampled i.i.d from $\Sigma(N_1^H,  S_k)$.
    In our MSFQ system, the first job that arrives in phase 1 needs to wait for $N_1^H$ heavy jobs to complete.
    In the EFS system, the second job in a busy period also needs to wait for $N_1^H$ heavy jobs' work to complete.
    The subsequent jobs in the busy period also see the same work in the queue as the subsequent jobs in phase 1 in our MSFQ system.
    In other words, a tagged job that arrives into the MSFQ system during phase 1 sees the same mean work compared to a job that does not receive exceptional service in the EFS system we consider.
    Hence, we can use the results in Remark~\ref{remark:ex-first} as if the exceptional job size distribution is $S'\sim \Sigma(N_1^H,  Exp(\mu_k))$ and the non-exceptional distribution is $S\sim Exp(\mu_k)$.
    Then,
    \begin{align}
        \E[T_1^H] &= \E[\textrm{mean work a tagged job sees}] + \frac{1}{\mu_k}\nonumber\\
        &=\E[W^{EFS}(\lambda_k, S_k, \Sigma(N_1^H,  S_k) \mid \textrm{no exceptional service})] + \frac{1}{\mu_k}
        =\frac{\E[W^{EFS}(\lambda_k, S_k, \Sigma(N_1^H,  S_k))]}{1-p^{EFS}(\lambda_k, S_k, \Sigma(N_1^H,  S_k))} + \frac{1}{\mu_k}. \label{eq:et1l}
    \end{align}

    The proof for $\E[T_2^L]$ follows a similar argument.
    We compare a tagged light job that arrives during phase 2 and an EFS system with an arrival rate $\lambda_1$, exceptional job size distribution $S'\sim \Sigma(N_2^L-k+1, S_1/k)$, and non-exceptional distribution $S\sim S_1/k$. Then,
    \begin{align}
        \E[T_2^L]&=\frac{\E[W^{EFS}(\lambda_1, S_1/k, \Sigma(N_2^L-k+1, S_1/k))]}{1-p^{EFS}(\lambda_1, S_1/k, \Sigma(N_2^L-k+1, S_1/k))} + \frac{1}{\mu_1}.\label{eq:t2s}
    \end{align}
   Based on \Cref{remark:ex-first} and \eqref{eq:et1l}, we can see that our $\E[T_1^H]$ formula depends on the first and second moments of $S_k$ and $\Sigma(N_1^H,  S_k)$.
    It is easy to see that the first and second moments of $S_k$ is $\E[S_k]=1/\mu_k$ and $\E[S_k^2]=2/\mu_k^2$.
    We can compute $\E[\Sigma(N_1^H, S_k)]=\E[N_1^H]/\mu_k$ because $N_1^H$ is independent from the job sizes.
    Lastly, to compute $\E[\Sigma(N_1^H,  S_k)^2]$, we have
    \begin{align*}
        \E[\Sigma(N_1^H,  S_k)^2]&=\sum_{i=0}^\infty \E[(\Sigma(i,  S_k)^2)]P\set{N_1^H=i}
        % =\sum_{i=0}^\infty (Var(\Sigma(i,  S_k))+\E[\Sigma(i,  S_k)]^2) P\set{N_1^H=i}\\
        =\sum_{i=0}^\infty \frac{i+i^2}{\mu_k^2} P\set{N_1^H=i}
        =\frac{\E[(N_1^H)^2]+\E[N_1^H]}{\mu_k^2}.
    \end{align*}
    Similarly, it suffices to compute first and second moments of $S_1/k$ and $\Sigma(N_2^L-k+1, S_1/k)$ to compute $\E[T_2^L]$.
    It is easy to see that $\E[S_1/k]=1/(k\mu_1)$ and $\E[(S_1/k)^2]=2/(k\mu_1)^2$. Likewise, we can compute $\E[\Sigma(N_2^L-k+1, S_1/k)]=(\E[N_2^L]-k+1)/(k\mu_1)$. Lastly, to compute $\E[\Sigma(N_2^L-k+1, S_1/k)^2]$, we have
    \begin{align*}
        \E[\Sigma(N_2^L-k+1, S_1/k)^2]
        %  \E[((N_2^L-k+1)\square (S_1/k))^2]
         % &=\frac{\E[(N_2^L-k+1)^2]+\E[N_2^L-k+1]}{k^2\mu_1^2}\\
         % &=\frac{\E[(N_2^L)^2+(k-1)^2-2N_2^L(k-1)+N_2^L-k+1]}{k^2\mu_1^2}\\
         % &=\frac{\E[(N_2^L)^2]+\E[(3-2k)N_2^L+k^2-2k+1-k+1]}{k^2\mu_1^2}\\
         &=\frac{\E[(N_2^L)^2]-(2k-3)\E[N_2^L]+k^2-3k+2}{k^2\mu_1^2}.
    \end{align*}
    Therefore, we have shown that it suffices to compute the first and second moments of $N_1^H$ and $N_2^L$ to compute $\E[T_1^H]$ and $\E[T_2^L]$.
\end{proof}

    % \zhongrui{Fix $Exp(\mu)$ notations and expand the symmetric argument for $\E[T_2^L]$}

    % \zhongrui{Should we expand the proof for lights?}

We now analyze the mean response time over all heavy jobs that arrive in any of phases 2, 3, or 4,
and the mean response time over all light jobs that arrive in either phases 1 or 4.
\begin{restatable}{lemma}{rexcess}
\label{thm:excess}
    The mean response time $\E[T_{2,3,4}^H]$
    of heavy jobs which arrive during phases 2, 3, or 4,
    and the mean response time $\E[T_{1,4}^L]$
    of light jobs which arrive during phases 1 or 4, are given by:
    \begin{align*}
        \E[T_{2,3,4}^H] = \frac{(\lambda_k/\mu_k + 1)\E[(\phase_2+\phase_{3}+\phase_4)^2]}{2\E[\phase_2+\phase_{3}+\phase_4]} + \frac{1}{\mu_k},
        \qquad 
        \E[T_{1,4}^L] = \frac{(\lambda_1/(k\mu_1) + 1) \E[(\phase_4+\phase_1)^2]}{2\E[\phase_4+\phase_1]} + \frac{1}{\mu_1}.
    \end{align*}
\end{restatable}
\begin{proof}
    Consider a tagged heavy job that arrives into the system in any of phases 2, 3, or 4.
    On arrival, this tagged job sees all the heavy job arrivals between the start of the most recent phase 2 and the current time in the system.
    Hence, the tagged job needs to wait until these jobs are completed in the upcoming phase 1
    before it can begin receiving service.
    Let $\phase_a^H$ be the \emph{age} of the $\phase_2+\phase_{3}+\phase_4$ period, the elapsed time from the beginning of phase 2 an arrival during this period sees, and $\phase_e^H$ be the \emph{excess} of the $\phase_2+\phase_{3}+\phase_4$ period, the time until the end of phase 4 an arrival during this period sees.
    On average, the tagged heavy job sees $\lambda_k \E[\phase_a^H]$ heavy jobs in the system at the time it arrives into the system.
    Hence, the mean work the heavy job sees is $\frac{\lambda_k}{\mu_k}\E[\phase_a^H]$.

    The response time of this tagged heavy job is composed of 3 elements: the time before any heavy job receives service, the mean work the tagged job sees in the system, and the time to serve itself. Hence,
    % Note that, before these tagged heavy jobs receive service, they all have to wait for an excess of the $\phase_2+\phase_3+\phase_4$ period because the system is exclusively serving light jobs.
    % Combining these time the tagged job needs to wait, and the average time it takes to serve the tagged job itself, we get 
    \begin{align}
       \E[T^H_{2,3,4}] &= \E[\phase_e^H]+\frac{\lambda_k}{\mu_k} \E[\phase_a^H]+\E[S_k]\label{eq:tl234}.
    \end{align}

    Combining the Palm inversion \cite{leboudec2010performance}, standard results on ages and excess~\cite{harchol2013performance}, and \eqref{eq:tl234}, we have:
    \begin{align*}
        \E[T^H_{2,3,4}] &= \frac{\E[(\phase_2+\phase_3+\phase_4)^2]}{2\E[\phase_2+\phase_3+\phase_4]}+\frac{\lambda_k}{\mu_k}\frac{\E[(\phase_2+\phase_{3}+\phase_4)^2]}{2\E[\phase_2+\phase_{3}+\phase_4]}+\frac{1}{\mu_k}\\
        &=\frac{(\lambda_k/\mu_k + 1)\E[(\phase_2+\phase_{3}+\phase_4)^2]}{2\E[\phase_2+\phase_{3}+\phase_4]} + \frac{1}{\mu_k}.
    \end{align*}

    Similarly, consider a tagged job that arrives into the system in either phase 4 or phase 1.
    On arrival, this tagged job sees all of the light job arrivals between the start of the most recent phase 4 and the current moment.
    let $H_a^L$ and $H_e^L$ be the age and excess of this $H_4+H_1$ period.
    The tagged light job that arrives during phase 4 or phase 1 also needs to wait for an excess and $\frac{\lambda_1}{k\mu_1}\E[H_a^L]$ of work before receiving service.
    Its response time is composed of the excess duration, the mean work it sees on arrival, and the time it takes to complete:
    \begin{align*}
        \E[T^L_{1,4}] &= \E[\phase_e^L]+\frac{\lambda_1}{k\mu_1} \E[\phase_a^L]+\E[S_1]
        =\frac{\E[(\phase_4+\phase_1)^2]}{2\E[\phase_4+\phase_1]} + \frac{\lambda_1}{k\mu_1}\frac{\E[(\phase_4+\phase_1)^2]}{2\E[\phase_4+\phase_1]} + \frac{1}{\mu_1} \\
        &=\frac{(\lambda_1/(k\mu_1) + 1) \E[(\phase_4+\phase_1)^2]}{2\E[\phase_4+\phase_1]} + \frac{1}{\mu_1}.
        \qedhere
    \end{align*}

    % \zhongrui{Do we need to cite some results on excess/age?}

    % Similarly, a tagged light job that arrives during phase 4 and phase 1 only needs to wait for arrivals 
\end{proof}
Therefore, we have shown that it suffices to compute the first and second moments of $H_i$ in order to compute $\E[T_{2,3,4}^H]$ and $\E[T_{4,1}^L]$.

We finish by analyzing the mean response $\E[T_3^L]$ of light jobs that arrive during phase 3.
\begin{restatable}{lemma}{r3s}
\label{thm:absorption}
The mean response time $\E[T_3^L]$ of light jobs that arrive during phase 3 is given by:
\begin{align*}
    \E[T_3^L]=\frac{\sum_{j=\ell+1}^\infty \frac{C_j}{\lambda_1+\min(k, j)\mu_1} \frac{k+(j-k+1)^+}{k\mu_1}}{\sum_{j=\ell+1}^\infty \frac{C_j}{\lambda_1+\min(k, j)\mu_1}},
\end{align*}
where $C_j$ is defined recursively as follows, for each positive integer $j$:
\begin{align*}
    C_j &=\begin{cases}
        \frac{\lambda_1+(\ell+1)\mu_1}{(\ell+1)\mu_1}\mathbbm{1}\set{\ell+1\le k-1} & j = \ell+1\\
        C_{j-1}\frac{\lambda_1(\lambda_1+j\mu_1)}{j\mu_1(\lambda_1+(j-1)\mu_1)}+\frac{\lambda_1+j\mu_1}{j\mu_1}\mathbbm{1}\set{j\le k-1} & \ell+1<j\le k\\
        \frac{\lambda_1}{k\mu_1} C_k & j>k
    \end{cases}.
\end{align*}
% \begin{align*}
%     \E[total\ response] &= \sum_{i=\ell+1}^{n-1} \frac{N_{n-1,i} }{\lambda_1 + i \mu_1} \frac{1}{\mu_1} +\frac{N_{n-1,n}}{\lambda_1+n\mu_1}\frac{1}{n \mu_1}\frac{1 + n(1-\lambda_1/(N\mu_1))}{(1-\lambda_1/(N\mu_1))^2}\\
%     \E[total\ time] &= \sum_{i=\ell+1}^{n-1} \frac{N_{n-1,i} }{\lambda_1 + i \mu_1} + \frac{N_{n-1,n}}{\lambda_1 + n \mu_1} \frac{1}{1-\lambda_1/(N\mu_1)}\\
%     \E[T_{3}^L]&=\frac{\E[total\ response]}{\E[total\ time]}
% \end{align*}
\end{restatable}
\begin{proof}    
    Consider the stochastic process $\set{n_1(t)}$ during phase 3. By the definition of MSFQ, $n_1=k-1$ at the beginning of phase 3 and $n_1=\ell$ at the end of phase 3.
    The process $\set{n_1(t)}$ forms an absorbing Markov chain corresponding to an $M/M/k$ system with arrival rate $\lambda_1$, and job size distribution $S_1$.
    In order to characterize the mean response time of a light job that arrives in phase 3, we condition the arrival based on the state it sees in $\set{n_1(t)}$.
    Therefore, it suffices to compute $C_j$, the number of visits to state $j$ in $\set{n_1(t)}$, to characterize the probability a light job arrives in state $j$.
    
    We first consider the case where there are at least $k+1$ light jobs in the system.
    When $j\ge k+1$, an arrival from state $j-1$ or $j$ will accrue one visit to state $j$, 
    \begin{align*}
        C_j&=C_{j-1}\frac{\lambda_1}{\lambda_1+k\mu_1}+C_j\frac{\lambda_1}{\lambda_1+k\mu_1}=C_{j-1}\frac{\lambda_1}{k\mu_1}.
    \end{align*}
    Therefore, for all $j\ge k+1$, we have $C_j=\frac{\lambda_1}{k\mu_1}C_k$.
    When $\ell+2\le j\le k$, an arrival from state $j-1$ or $j$ will accrue one visit to state $j$ with different rates:
    \begin{align*}
        C_j&=C_{j-1}\frac{\lambda_1}{\lambda_1+(j-1)\mu_1}+C_j\frac{\lambda_1}{\lambda_1+j\mu_1}+\mathbbm{1}\set{j\le k-1}=C_{j-1}\frac{\lambda_1(\lambda_1+j\mu_1)}{j\mu_1(\lambda_1+(j-1)\mu_1)}+\frac{\lambda_1+j\mu_1}{j\mu_1}\mathbbm{1}\set{j\le k-1}.
    \end{align*}
    Finally, we handle $C_{\ell+1}$ separately:
    \begin{align*}
        C_{\ell+1}&=C_{\ell+1}\frac{\lambda_1}{\lambda_1+(\ell+1)\mu_1}+\mathbbm{1}\set{\ell+1\le k-1}=\frac{\lambda_1+(\ell+1)\mu_1}{(\ell+1)\mu_1}\mathbbm{1}\set{\ell+1\le k-1}.
    \end{align*}
    To compute $\E[T_3^L]$, note that the response time of a light job that arrives during phase 3 depends only on the number of light jobs $n_1(t)$ when the light job arrives.
    A light job that sees $j$ jobs on arrival has expected response time $\frac{k+(j-k+1)^+}{k\mu_1}$.
    %Note also that during phase 3, the number of light jobs in the system evolves exactly according to the absorbing Markov chain defined above.
    As a result,
    we can compute $E[T_3^L]$ by conditioning on $n_1(t)$ seen on arrival. By PASTA, this is simply the time-average distribution of the number of jobs seen during phase 3, which is given by the pre-absorption average of $C_i$. 
    Specifically,
    \begin{align*}
        \E[T_3^L]&=\frac{\sum_{j=\ell+1}^\infty\E[\textrm{total time spent in state $j$}] \cdot  \E[T_3^L\mid \textrm{the job arrives during state $j$}]}{\sum_{j=\ell+1}^\infty\E[\textrm{total time spent in state $j$]}}\\
        &=\frac{\sum_{j=\ell+1}^\infty \frac{C_j}{\lambda_1+\min(k, j)\mu_1} \frac{k+(j-k+1)^+}{k\mu_1}}{\sum_{j=\ell+1}^\infty \frac{C_j}{\lambda_1+\min(k, j)\mu_1}}.
        \qedhere
    \end{align*}
\end{proof}
We have shown that to compute the $\E[T_i]$ terms and the $m_i$ terms
in \eqref{eq:er} for each phase $i$,
it suffices to compute the first and second moments of $\phase_i$ and $N_i$.
To compute the first and second moments,
it suffices in turn to compute the transforms of $\phase_i$ and $N_i$ for all $i$.
We tackle these transforms in the following section.
%Therefore, it suffices to derive transforms of $T$'s and $N$'s in order to estimate the mean response time of the system.
\subsection{Phase Duration Analysis}
\label{sec:cycle-analysis}
% (2-class MSFQ, including MSF)
% Up to this point, our analysis relating $\E[T]$ to the first and second moments of $\phase_i$ and $N_i$ has been exact.
% If we could exactly compute the transforms of $\phase_i$ and $N_i$ for all $i$, we would therefore have an exact analysis of $\E[R]$.
% While computing these transforms exactly is hard, we provide approximations of these transforms in Lemmas X-Y.
In this section, we compute the required transforms of $\phase_i$ and $N_i$ for all phases $i$ in order to compute their first and second moments.
Taken together with Lemmas \ref{lem:fraction}-\ref{thm:absorption}, this completes the proof of \cref{thm:summary}.
% we provide an approximation of $\E[T]$ which we show to be highly accurate via simulation in Section \ref{sec:sim}.

In our analysis, it will be useful to refer to \emph{busy periods} of the system when serving either light or heavy jobs.  We define these busy periods as follows.

\begin{restatable}[Busy Periods]{remark}{busyperiod}
    %As stated in~\cite{harchol2013performance}, 
    We consider a busy period started by a random amount of work $W$ to be the time required for an M/G/1 system to empty when starting with $W$ work in the system.
    We let $B^H_W$ be the duration of this busy period in an M/G/1 where only heavy jobs arrive, with arrival rate $\lambda_k$.
    Similarly, let $B^L_W$ be the duration of this busy period in an M/G/1 where only light jobs arrive, with arrival rate $\lambda_1$.
    %We also let $\wt{B^H_W}(s)$ and $\wt{B^L_W}(s)$ be the Laplace transforms of $B^H_W$ and $B^L_W$. 
    Using standard queueing-theoretic techniques~\cite{harchol2013performance}, we have
    \begin{align*}
        \wt{B^H_W}(s)&=\wt{W}(s+\lambda_k-\lambda_k \wt{B^H_{S_k}}(s)),
        \qquad 
    \wt{B^L_W}(s)=\wt{W}(s+\lambda_1-\lambda_1 \wt{B^L_{S_1}}(s)).
    \end{align*}
\end{restatable}

To begin, we observe that the duration of phase 1 is equal to a busy period started by the number of heavy jobs that arrive during the preceding phases 2-4.
Similarly, the duration of phase 2 is a busy period started by the number of light jobs that arrive during the preceding phases 4 and 1.
This insight allows us to analyze $\wt{\phase_1}(s)$ and $\wt{\phase_2}(s)$ in Lemma~\ref{lemma:cycle-length}.
\begin{restatable}{lemma}{cyclelength}
    \label{lemma:cycle-length}
     The transforms of the distributions of phase 1 and phase 2 durations are given by:
    \begin{align*}
        \wt{\phase_1}(s)&=\wh{N_1^H}(\wt{B^H_{S_k}}(s)),
        \qquad
        \wt{\phase_2}(s)=\wh{N_2^L}(\wt{B^L_{S_1}}(s))(\wt{B^L_{S_1}}(s))^{1-k}.
    \end{align*}
\end{restatable}
\begin{proof}
% The phase duration $\phase_1$ given there are $x$ class-$N$ jobs at the beginning of phase 1 can be written as the distribution of a busy period started by the work of $x$ class-$N$ jobs.
% The phase duration $\phase_2$ given there are $x$ class-$1$ jobs at the beginning of phase 2 can be written as the distribution of a busy period started by the work of $x-N+1$ class-$1$ jobs, when $x-N+1$ is nonnegative.
% Formally, 
% \begin{align*}
%     [\phase_1 \mid N^H_1 =x] &\sim B^H_{\Sigma(x,  S_N)},\\
%     [\phase_2 \mid N^L_2 =x] &\sim B^L_{(x-N+1)\square S_1}.
% \end{align*}
% Specifically, if phase 1 starts with $x$  class-$N$ jobs, then the phase duration of phase 1 is a busy period started by these $x$ class-$N$ jobs.
% More formally, we have $$[\phase_1 \mid N^H_1 =x]\sim B^H_{\Sigma(x,  S_N)}.$$

% Similarly, we have $$
%     [\phase_2 \mid N^L_2 =x] \sim B^L_{(x-N+1)\square S_1},
% $$
% because the phase duration of phase 2 is a busy period started by $x-N+1$ light jobs for some $x>N-1$.
% Note that if $x\le N-1$, the system will run in phase 3 if $x\le \ell$ and phase 4 if $x<\ell$.
% From the phase duration results, we have $[\phase_1 \mid N^H_1 =x]\sim B^H_{\Sigma(x,  S_N)}$.
% By unconditioning $N_1^H=x$, we get $\phase_1\sim B^H_{\Sigma(N_1^H,  S_N)}$. Therefore,
At the beginning of phase 1, the system has $\Sigma(N_1^H,  S_k)$ amount of work in terms of heavy jobs.
At the end of phase 1, the system empties all the heavy jobs in the system.
Therefore, the length of phase 1 can be seen as a busy period for heavy jobs started by $\Sigma(N_1^H,  S_k)$ amount of work.
\begin{align*}
    \wt{\phase_1}(s)=\wt{B^H_{\Sigma(N_1^H,  S_k)}}(s)
    % &=\wt{\Sigma(N_1^H,  S_N)}(s+\lambda_N-\lambda_N\wt{B^H_{S_N}}(s))\\
    =\wh{N^H_1}(\wt{S_k}(s+\lambda_k-\lambda_k \wt{B^H_{S_k}}(s)))
    =\wh{N_1^H}(\wt{B^H_{S_k}}(s))\,.
\end{align*}
Here, we use the fact that $\wt{\Sigma(X, Y)}(s) = \wh{X}(\wt{Y}(s))$ where $X$ is either independent of $Y$ or is a stopping time relative to the sequence of durations $Y$ \cite{harchol2013performance}.

Similarly, the system has $N_2^L$ light jobs at the beginning of phase 2 and will have $k-1$ light jobs at the end of phase 2.
In this case, the system finishes $N_2^L-k+1$ jobs over the course of phase 2.
Therefore, the length of phase 2 can be seen as a busy period for light jobs started by $\Sigma(N_2^L-k+1, S_k)$ amount of work.
\begin{align*}
    \wt{\phase_2}(s)&=\wt{B^L_{\Sigma(N_2^L-k+1, S_k)}}(s)
    % &=\wt{(N_2^L-N+1)\square S_N}(s+\lambda_1-\lambda_1\wt{B^L_{S_1}}(s))\\
    % &=\wh{N_2^L-N+1}(\wt{S_1}(s+\lambda_1-\lambda_1\wt{B^L_{S_1}}(s)))\\
    =\wh{N_2^L-k+1}(\wt{B^L_{S_1}}(s))
    =\wh{N_2^L}(\wt{B^L_{S_1}}(s))(\wt{B^L_{S_1}}(s))^{1-k}\,.
    % &=\wh{N_2^L}(\wt{S_1}(s+\lambda_1-\lambda_1\wt{B^L_{S_1}}(s)))(\wt{S_1}(s+\lambda_1-\lambda_1\wt{B^L_{S_1}}(s)))^{1-n}
    \qedhere
\end{align*}
\end{proof}

Next, we compute the z-transforms of the number of jobs in the system at the start of each phase.  We note that our response time analysis depends only on the moments of $N_1^H$ and $N_2^L$, hence it suffices to compute z-transforms $\wh{N_1^H}(z)$ and $\wh{N_2^L}(z)$.
We show how these transforms depend on the Laplace transforms of the phase durations in Lemma~\ref{lem:ztrans}.

\begin{restatable}{lemma}{numberphases}
\label{lem:ztrans}
    The z-transforms of the distributions of the number of heavy jobs at the beginning of phase 1 and the number of light jobs at the beginning of phase 2 are given by:
    \begin{align*}
        \wh{N_1^H}(z)&=\wt{\phase_2}(\lambda_k(1-z))\wt{\phase_{3}}(\lambda_k(1-z))\wt{\phase_{4}}(\lambda_k(1-z))\\
        \wh{N_2^L}(z)&=\wt{\phase_2}(\lambda_k(1-\beta(z)))\wt{\phase_{3}}(\lambda_k(1-\beta(z)))\wt{\phase_{4}}(\lambda_k(1-\beta(z))+\lambda_1(1-z)),\\
        \text{where } \beta(z)&= \wt{B^H_{S_k}}(\lambda_1(1-z)).
    \end{align*}
\end{restatable}
\begin{proof}
The number of heavy jobs at the beginning of phase 1 can be seen as the number of arrivals accrued during a $\phase_2+\phase_{3}+\phase_4$ time period because the heavy jobs are emptied at the end of phase 1.
Formally, we have $N^H_1\sim A^H_{\phase_2+\phase_{3}+\phase_4}$, where $A^H_X$ denotes the number of heavy job arrivals in $X$ seconds.

Similarly, the number of light jobs at the beginning of phase 2 can be seen as the number of arrivals accrued during phase 4 and phase 1.
However, note that phases 4 and 1 are not independent, as they are positively correlated.
Intuitively, a longer phase 4 will result in a longer phase 1 because of more heavy jobs arriving into the system.
In this case, we use $\phase_{4,1}$ to denote the length of the joint phase 4 and phase 1 period.
As a result, we have $N^L_2\sim A^L_{\phase_{4,1}}$.
Note that, if we are only interested in the mean of this joint period, $\E[\phase_{4,1}]=\E[\phase_4]+\E[\phase_1]$ due to the linearity of expectations.

Next, we use the standard transform formulas for Poisson arrivals during a random interval: $\wh{A^H_X}(z)=\wt{X}(\lambda_k (1-z)), \wh{A^L_X}(z)=\wt{X}(\lambda_1 (1-z))$ \cite{harchol2013performance}.
Plugging into $N_1^H$, we have
\begin{align*}
    \wh{N_1^H}(z)&=\wh{A_{H_2+H_3+H_4}^H}(z)
    =\wt{H_2+H_3+H_4}(\lambda_k(1-z))
    =\wt{\phase_2}(\lambda_k(1-z))\wt{\phase_{3}}(\lambda_k(1-z))\wt{\phase_{4}}(\lambda_k(1-z)).
\end{align*}

Plugging into $N_2^L$, we have
\begin{align*}
    \wh{N_2^L}(z)&=\wh{A_{H_{4,1}}^L}(z)=\wt{H_{4,1}}(\lambda_1(1-z))
    =\int_0^\infty \wt{H_{4,1}}(\lambda_1(1-z)\mid H_4=x)P\set{H_4=x} dx\\
    % &=\int_0^\infty \wt{H_1}(\lambda_1(1-z)\mid H_4=x)e^{-\lambda_1(1-z)x}P\set{H_4=x} dx\\
    &=\int_0^\infty \wh{N_1^H}(\beta(z)\mid H_4=x)e^{-\lambda_1(1-z)x}P\set{H_4=x} dx\\
    % &=\int_0^\infty \wt{H_2}(\lambda_N(1-\beta(z)))\wt{H_3}(\lambda_N(1-\beta(z)))\wt{x}(\lambda_N(1-\beta(z))) e^{-\lambda_1(1-z)x}P\set{H_4=x}dx\\
    % &=\wt{H_2}(\lambda_N(1-\beta(z)))\wt{H_3}(\lambda_N(1-\beta(z)))\int_0^\infty e^{-\lambda_N(1-\beta(z))x}e^{-\lambda_1(1-z)x}dx\\
    &=\wt{\phase_2}(\lambda_k(1-\beta(z)))\wt{\phase_{3}}(\lambda_k(1-\beta(z)))\wt{\phase_{4}}(\lambda_k(1-\beta(z))+\lambda_1(1-z)).
    \qedhere
\end{align*}
% By definition, phase 3 starts with $N-1$ light jobs and ends when there are $\ell$ light jobs left.
% MSFQ still serves new light job arrivals in phase 3.
% Phase 4 starts when there are $\ell$ light jobs and ends when all these $\ell$ light jobs complete service.
% MSFQ does not serve new light job arrivals in phase 4.
\end{proof}

% We assume there are always at least $N$ jobs at the beginning of phase 2 to simplify our analysis.
% We will show that our cycle analysis, throughput analysis, and response time analysis results still hold given this assumption and will measure the effect of this approximation in our evaluation sections.

% TODO: Do we want to put N's given T's here?
Compared to $\phase_1$ and $\phase_2$, the durations $\phase_3$ and $\phase_4$ are relatively straightforward to analyze.
Specifically, the length of phases 3 and 4 depends only on the definition of the number of servers, $k$, and MSFQ threshold, $\ell$.
These phases are independent of the lengths of other prior phases.
We compute $\wt{\phase_3}(s)$ and $\wt{\phase_4}(s)$ in Lemma~\ref{lem:h3} and Lemma~\ref{lem:h4}, deferring the proofs to \ref{app:h3-h4}:
\begin{restatable}{lemma}{tthreeperiod}
\label{lem:h3}
    The Laplace transform of the duration of phase 3 is given by:
    $
        \wt{\phase_3}(s)=\prod\nolimits_{j=\ell+1}^{k-1} \wt{\phase_{3, j}}(s),
    $
    where $H_{3, j}$ is the transit time from $j$ light jobs in the system to $j-1$ light jobs in the system, with transform:
    \begin{align}
        \wt{\phase_{3, j}}(s)&=\begin{cases}
            \frac{j\mu_1}{\lambda_1+j\mu_1+s-\lambda_1 \wt{\phase_{3, j+1}}(s)} & j < k\\
            \wt{B^L_{S_1}}(s) & j = k\,.
        \end{cases}\label{eq:h3comp}
    \end{align}
\end{restatable}
\begin{comment}    
\begin{proof}
% Let $\phase_{3, j}$ be the random variable that denotes the period in phase 3 that starts when the system has $i$ light jobs and ends when the system first has $i-1$ light jobs.
To characterize $\phase_{3, j}$, we condition on the first event that happens during the $\phase_{3, j}$ period.
When the first event is an arrival, the remainder of $\phase_{3, j}$ consists of a $\phase_{3, j+1}$ period, followed by another independent $\phase_{3, j}$ period.
When the first event is a completion, the period ends.
In particular, we have 
\begin{align*}
    \phase_{3, j}&=\begin{cases}
        Exp(j\mu_1+\lambda_1)+\phase_{3, j+1}+\phase_{3, j} & \textrm{next event is an arrival}\\
        Exp(j\mu_1+\lambda_1) & \textrm{next event is a departure}
    \end{cases}.
\end{align*}
Note that $\phase_{3, k}\sim B^L_{S_1}$ as the time going from having $k$ to $k-1$ light jobs is the busy period because the completion rate of light jobs here is $n\mu_1$.
Then, the phase duration of phase 3 is the sum of these periods. Formally,
$\phase_{3}=\sum\nolimits_{j=\ell+1}^{n-1}\phase_{3, j}.$
Then, by standard transform techniques, \cref{lem:h3} holds.
\end{proof}
\end{comment}

\begin{restatable}{lemma}{tfourperiod}
\label{lem:h4}
The Laplace transform of the duration of phase 4 is given by:
$
    \wt{H_4}(s)=\prod\nolimits_{j=1}^\ell \frac{j\mu_1}{j\mu_1+s}.
$ 
\end{restatable}
\begin{comment}    
\begin{proof}
    In phase 4, no further light job arrivals are allowed into service, and there are $\ell$ light jobs at the beginning of phase 4.
    Therefore, we can write $\phase_4$ as a sum of i.i.d. exponential distributions: $\phase_4=\sum_{j=1}^\ell Exp(j\mu_1).$
    Therefore, by standard transform techniques,
    $
        \wt{H_4}(s)
        =\prod_{j=1}^\ell \wt{Exp(j\mu_1)}(s)=\prod_{j=1}^\ell \frac{j\mu_1}{j\mu_1+s}.
    $
\end{proof}
\end{comment}

Given Lemmas \ref{lemma:cycle-length}-\ref{lem:h4}, we are now ready to prove the summary theorem, Theorem~\ref{thm:summary}.

\msfqsummary*

\begin{proof}[Proof of \Cref{thm:summary}]
    Theorem \ref{thm:summary} follows from Lemmas \ref{lemma:cycle-length}-\ref{lem:h4}.
    These lemmas provide recursively-defined transforms that can be differentiated to obtain $\E[T]$.
    For convenience, we provide a calculator that performs these computations and returns the desired approximation
    \footnote{The calculator program can be found at \href{https://github.com/jcpwfloi/msfq-calculator}{https://github.com/jcpwfloi/msfq-calculator}}.
    % Lemmas \ref{thm:exceptional-first}-\ref{thm:absorption} show that the mean response time of an MSFQ policy, $\E[T]$, depends on the system parameters ($k, \lambda_1, \lambda_k, \mu_1, \mu_k$) and the first and second moments of $\phase_i$ and $N_i$ for all phases $i$.
    % Lemmas \ref{lemma:cycle-length}-\ref{lem:h4} show how to compute the transforms of $\phase_i$ and $N_i$ that suffice to compute these moments.
    % Specifically, while the formulas in Lemma \ref{lemma:cycle-length}-\ref{lem:h4} are recursive, the necessary moments can be computed by differentiating these formulas and solving a system of X equations.
    % Plugging these moments into \eqref{eq:er} and \cref{lem:fraction} yields an approximation of $\E[T]$ as desired.

\end{proof}

\section{Simulation Results}
\label{sec:sim}

%\zhongrui{Minor edits on consistency on language and formulas}
%\zhongrui{Adding connective tissues between 6.1, 6.2, and 6.3}

%\zhongrui{Plots: 
%Figure 2a and figure 2c should go side-by-side.\\
%Figure 2b should be splitted into two plots --- one for small jobs and one for heavy jobs, following the same style as 2a.\\
%Figure 3a should only show combined cycle lengths instead by phases breakdown.\\
%Figure 3b: y-axis should start from 0.
%}
%\textbf{Izzy: Only show raw response time once, to introduce the concept of weighted response time. Subsequently, only weighted response time.}

In \Cref{sec:results}, we derived an approximation of the mean response time under an MSFQ policy in the one-or-all case.
This raises two important questions.
First, how does MSFQ compare to other non-preemptive scheduling policies in the one-or-all case?
And second, how do these results generalize to workloads with additional classes of jobs?
To address these questions, we now evaluate MSFQ and several policies from the literature, comparing the mean response time predicted by our theoretical results to simulations of various other policies.
Specifically, we compare MSF against the First-Fit BackFilling policy \cite{grosof-2023-serverfilling}, and against the nonpreemptive Markovian Service Rate policy (nMSR) \cite{chen2025improving}.
We will show that our approximations from \Cref{sec:results} are highly accurate, and that MSFQ significantly outperforms all competitor policies.
% For completeness, we will also consider the preemptive ServerFilling policy, whose performance will demonstrate the power of preemption, to be weighted against the implementation difficulties that were mentioned before.

We begin by simulating policies in the one-or-all case in \Cref{sec:sim-one-or-all} to show that our response time analysis is accurate and that MSFQ is by far the best of the non-preemptive scheduling policies.
We then study two natural generalizations of MSFQ to workloads with more than two job classes.
We show that these generalizations, Adaptive Quickswap and Static Quickswap, perform well under more general workloads using both synthetic traces (\cref{sec:sim-general}) and traces from the Google Borg cluster scheduler (\cref{sec:sim-borg}).

For simulation results, we wrote a discrete event simulation framework specifically developed for MSJ systems. Our simulator implements a wide range of scheduling policies and can either generate synthetic workloads or use real-world traces.
This framework is available on GitHub\footnote{\href{https://github.com/NeDS-Lab/mjqm-simulator/}{https://github.com/NeDS-Lab/mjqm-simulator/}}.

\subsection{Simulation Metrics}
\label{sec:weighted-metric}
To evaluate our simulations, we will use a variety of response time metrics.
For each class of jobs, $j$, we define $\E[T^{(j)}]$ to be the mean response time of class-$j$ jobs.
This allows us to examine the mean response time of each class separately to see how a policy balances the response times between job classes.

Assuming a workload consisting of $m$ job classes, we can write the mean response time across all jobs as
$$\E[T] = \sum_{j=1}^m p_j \E[T^{(j)}].$$

We note, however, that as the number of job classes grows and the server needs and job sizes vary more between job classes, $\E[T]$ is not always the most meaningful metric.
Specifically, in more complex cases, a large fraction of the system load can be composed of a small fraction of the jobs in the system.
These jobs, which generally have large server needs and large mean job sizes, will be mostly ignored in the computation of $\E[T]$ because their $p_i$ terms are small.
For example, we find that in the Google Borg cluster scheduler, 85.8\% of the system load is contributed by just 0.34\% of the jobs in the workload.
This aligns with the findings of the original Google Borg papers \cite{borgClusterManager,tirmazi2020borg}.

To illustrate the problem with mean response time in this scenario, we examine the fairness properties of various scheduling policies in \DIFnomarkup{\ref{app:fairness}}.
We find that policies that appear to perform well with respect to mean response time actually allow the heavy jobs in the system to suffer disproportionately.
For example, under MSF, the mean response time of the heaviest jobs can be several orders of magnitude larger than the mean response time of the other job classes, even though the overall mean response time remains low.
Because the heavy jobs can comprise a significant fraction of the overall system load (and therefore a large portion of the revenue for a system operator), this degree of unfairness cannot be tolerated.

It is more realistic to \emph{balance} overall mean response time and fairness.
We therefore introduce \emph{weighted mean response time}, a metric that allows us to consider mean response time and fairness simultaneously.
We define weighted mean response time as
\begin{align*}
E[T^w] = \frac{\sum_{j=1}^k j/\mu_j\cdot p_j E[T^{(j)}]}{\sum_{i=1}^k i/\mu_i\cdot p_i}=\sum_{j=1}^k\frac{ \rho_j }{\rho}E[T^{(j)}],
\end{align*}
where $\rho_j=\frac{j\lambda_j}{\mu_j}$ is the system load contributed by class-$j$ jobs and $\rho=\sum_{j=1}^k \rho_j$ is the total system load.
Under this definition, each job class's weight corresponds to the fraction of load it contributes to the system.
Said another way, a class's weight is proportional to the server-hours used by the class (cost paid), preventing the scenario where a scheduling policy can ignore the infrequent but heavy jobs in the workload.

% To evaluate the propodurationssed scheduling policies {and to verify the accuracy of our approximations}, we conduct simulation experiments, estimating mean response time and system stability under variable configurations. We use a custom C++ discrete event simulation framework specifically developed for the multiserver-job model (MSJ). The simulator implements a wide range of scheduling policies and can represent different system configurations and workloads. 
% The simulator has been previously validated in simpler settings by comparing its outcomes against approximate and analytical models~(removed for anonymity). 

% We begin in \cref{sec:sim-one-or-all} by presenting results for the the Most Servers First (MSF) and MSF Quickswap (MSFQ) in the one-or-all MSJ setting. We demonstrate the accuracy of our analytical mean response time characterization, \cref{thm:summary}.
% In \cref{sec:sim-general}, we consider heavyr numbers of job classes.
% Finally, in \cref{sec:sim-borg}, we evaluate the scheduling policies on more complex and realistic workloads derived from the Google Borg data center trace data \cite{tirmazi2020borg}.

\subsection{Two job classes: one-or-all MSJ}
\label{sec:sim-one-or-all}

%\zhongrui{Note that under our metric nMSR is worse, but under E[T] metric, nMSR is better compared to MSF}

We first evaluate the performance of MSFQ with $\ell=k-1$ in the one-or-all setting analyzed in \Cref{sec:results}.
We compare MSFQ to MSF, as well as the First-Fit and nMSR policies examined in the prior work.
Our simulations consider a system with $k=32$ servers, where 90\% of job arrivals are light jobs and the mean job size is 1 for both heavy and light jobs.
That is, $p_1=0.9$, $p_k=0.1$, and $\mu_1=\mu_k=1$.
These parameters reflect the common setting where 10\% of the jobs (the heavy jobs) comprise about 80\% of the load on the system.
We set $\ell=k-1$ because all servers will be utilized when there are $k$ or more light jobs in the system.
As soon as there are fewer than $k$ light jobs in the system and some servers are idle, it makes sense to try to serve the heavy jobs in the system.
While the exact choice of $\ell$ does affect system performance, Figure \ref{fig:threshold} shows that mean response time is largely independent of $\ell$ as long as $\ell$ is not set very close to 0.
\begin{wrapfigure}{r}{0.48\textwidth}
 \vspace{-.2in}
  % \begin{center}
    \includegraphics[width=.48\textwidth]{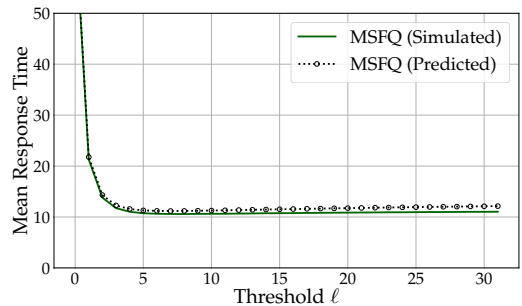}
  % \end{center}
    % \vspace{-.15in}
    \caption{Impact of the threshold value, $\ell$, on mean response time of the MSFQ policy evaluated in \Cref{fig:rt_oneorN_32}.}\label{fig:threshold}
\end{wrapfigure}

\begin{figure}[t]
\centering
\begin{subfigure}[t]{0.48\textwidth}
    \centering
    \includegraphics[width=\textwidth]{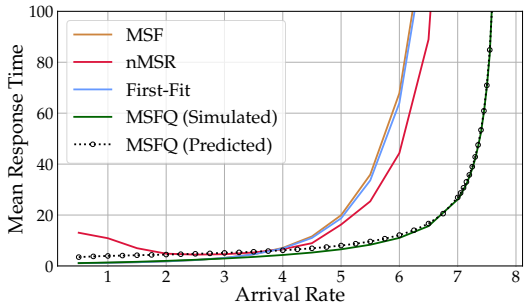}
    \caption{Overall unweighted mean response time}\label{fig:rt_oneorN}
\end{subfigure}\quad
\begin{subfigure}[t]{0.48\textwidth}
    \centering
    \includegraphics[width=\textwidth]{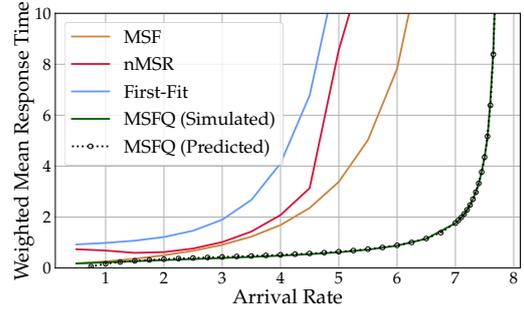}
    \caption{Overall weighted mean response time} \label{fig:rt_oneorN_weighted}
\end{subfigure}\quad
\begin{subfigure}[t]{0.48\textwidth}
    \centering
    \includegraphics[width=\textwidth]{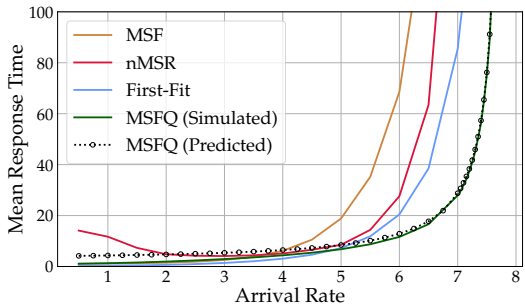}
    \caption{Mean response time of light jobs} \label{fig:rt_oneorN_small}
\end{subfigure}\quad
\begin{subfigure}[t]{0.48\textwidth}
    \centering
    \includegraphics[width=\textwidth]{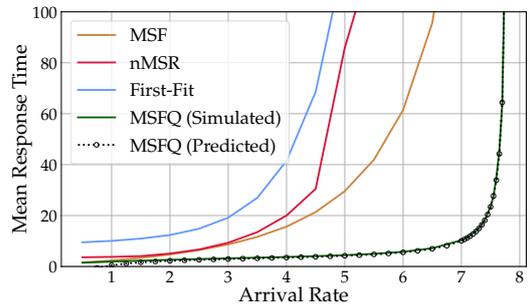}
    \caption{Mean response time of heavy jobs} \label{fig:rt_oneorN_heavy}
\end{subfigure}\quad
\caption{Mean response time as a function of job arrival rate in a one-or-all MSJ system with $k=32$, $p_1=0.9$, and $\mu_1 = \mu_k = 1$.  MSFQ beats all other non-preemptive policies in terms of both mean response time and weighted mean response time.  In particular, MSFQ can be two orders of magnitude better than MSF and nMSR with respect to both metrics.}
\label{fig:rt_oneorN_32}
\vspace{-0.2in}
\end{figure}

\begin{wrapfigure}{r}{0.48\textwidth}
\vspace{-.4in}
  % \begin{center}
    \includegraphics[width=.48\textwidth]{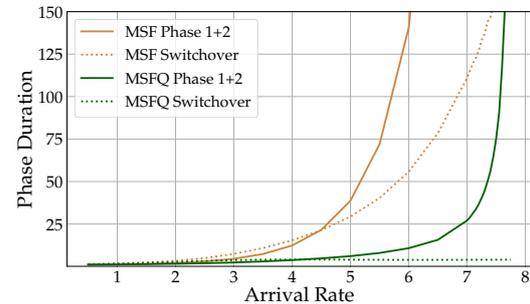}
  % \end{center}
  % \vspace{-.15in}
    \caption{Service phase durations for the MSFQ policy evaluated in \Cref{fig:rt_oneorN_32}.}\label{fig:phases}
    \vspace{-.15in}
\end{wrapfigure}

\cref{fig:rt_oneorN_32} shows the effect of varying the arrival rate, $\lambda$, on weighted and unweighted mean response time. 
%\cref{fig:rt_oneorN_weighted} across all jobs in the system under our MSFQ policy with threshold parameter $\ell=k-1$, and under the MSF policy
We see that our analysis of the mean response time under MSFQ is highly accurate at a wide range of arrival rates.
Furthermore, our MSFQ policy achieves the best weighted and unweighted mean response time in all cases, outperforming the competitor policies by two orders of magnitude when the arrival rate is high.
We also measure the mean response time for each job class in \cref{fig:rt_oneorN_small,fig:rt_oneorN_heavy}.
These results confirm that MSFQ improves both classes' mean response time individually to improve the overall mean response time.
% Note that the choice of weighted versus unweighted mean response time impacts the ordering of the competitor policies.
% The nMSR and First-Fit policies benefit in the unweighted case by prioritizing light jobs, while MSF benefits in the weighted case by prioritizing heavy jobs.
% In both cases, none of the baseline policies can match MSFQ.

% \begin{figure}[H]
% \centering
% \begin{subfigure}[t]{0.42\textwidth}
%    \centering
%    \includegraphics[width=\textwidth]{newplots/threshold.pdf}
%    \caption{Impact of the threshold value, $\ell$, on mean response time of the MSFQ policy evaluated in \Cref{fig:rt_oneorN_32}.}\label{fig:threshold}
% \end{subfigure}
% \quad
% \begin{subfigure}[t]{0.42\textwidth}
%    \centering
%    \includegraphics[width=\textwidth]{newplots/phases.pdf}
%    \caption{Service phase durations for the MSFQ policy evaluated in \Cref{fig:rt_oneorN_32}.}\label{fig:phases}
% \end{subfigure}
% \caption{}
% \end{figure}

Our response time analysis in \cref{sec:rt} showed that MSFQ improves mean response time by switching faster between service phases.
To illustrate this, we measure the phase durations of MSF and MSFQ during the above simulations in \cref{fig:phases}.
Recall that MSF is equivalent to an MSFQ policy with threshold $\ell=0$, and the MSFQ policy in this case uses $\ell=k-1$.
Hence, both policies have full resource utilization in phases 1 and 2, and use the remaining phases to switch the class of job in service.
\cref{fig:phases} shows that MSFQ has shorter switching phases, leading to much shorter durations of phases 1 and 2.

We further illustrate the impact of having shorter phase durations in \cref{fig:threshold}, which shows the effect of the threshold value, $\ell$, on the mean response time of MSFQ.
Using any threshold value larger than 0 has a dramatic benefit on mean response time by allowing faster switchover times and shorter phase durations.
We also note that, while setting an arbitrarily large threshold could waste capacity by causing the system to switch too frequently, this effect is limited in practice. %\textcolor{red}{We have no data to support the claim above. We did not try threshold values heavyr than k -- we tried 2k and results are marginally worse than with k }. 
% \textcolor{blue}{Additionally, we see from \cref{fig:threshold} that all values between 4 and 31 lead to similar results for the mean job response time, indicating that the coice of the value of $\ell$ is not critical, as we already noted.}
Hence, while our theoretical results can be used to select the optimal value of $\ell$, a good heuristic appears to be to choose $\ell = k - 1$.

\subsection{Generalizing to Additional Job Classes}
\label{sec:sim-general}

While MSFQ has good performance in the one-or-all case, it is not immediately clear how these results generalize to cases with additional job classes.
Hence, we now simulate the Static Quickswap and Adaptive Quickswap policies defined in \Cref{sec:policies}.
These policies generalize MSFQ to cases with many job classes, using the Quickswap mechanism to try and maintain the short phase durations of MSFQ.
Given the added variability in server needs, we will focus on weighted mean response time in this multiclass case.
We consider a system with $k=15$ servers and 4 classes: class-1, class-3, class-5, and class-15.
Each class has a mean job size of 1, and we set $p_1=0.5$, $p_3=0.25$, $p_5=0.2$, and $p_{15}=0.05$. 
Note that we have chosen the server needs to divide $k$ so that any one class of jobs can utilize all $k$ servers.
By \Cref{rem:general-stability}, it is therefore possible to stabilize the system when $\lambda < 5$.
% We first consider a system with $k=8$ servers and 4 classes: class-1, class-2, class-4, and class-8 jobs.
% Each class of job has a mean job size of 1, and we set $p_1=0.4$, $p_2=0.3$, $p_4=0.2$, and $p_8=0.1$.
% Here, both $k$ and the server needs of every class are chosen to be powers of 2.
% In this case, the system can achieve complete server utilization as long as there are at least 8 jobs in the system \cite{grosof-2023-serverfilling}.
%TODO: clarify stability?
% \Cref{rem:general-stability}, the system is stable when $\lambda < 3.07$.
% We first look at a MSJ system with 8 servers and 4 classes, all with server needs that are powers of two: 1,~2,~4,~and 8. All jobs classes have $\mu_i = 1$.
% The occurrence probabilities of job classes are $[p_1,p_2,p_4,p_8]=[0.4, 0.3, 0.2, 0.1]$.

% \begin{figure}[t]
% \centering
% \begin{subfigure}[t]{0.48\textwidth}
%     \centering
%     \includegraphics[width=\textwidth]{newplots/4classp2.pdf}
%     \caption{Weighted mean response time (power of 2)}\label{fig:pof2}
% \end{subfigure}
% \quad
% \begin{subfigure}[t]{0.48\textwidth}
%     \centering
%     \includegraphics[width=\textwidth]{newplots/4classnop2.pdf}
% \caption{Weighted mean response time (non power of 2)}\label{fig:notpof2}
% \end{subfigure}\quad
% \caption{Weighted mean response time general settings}
% \label{fig:general_msf}
% \end{figure}

\Cref{fig:pof2} shows that Static and Adaptive Quickswap both provide an advantage over the competitor policies. Adaptive Quickswap performs the best in practice because it uses more complex switching logic to avoid having unused servers. Static Quickswap performs slightly worse than Adaptive Quickswap in all cases. However, Static Quickswap is guaranteed to be throughput-optimal by \Cref{rem:general-stability}, while Adaptive Quickswap has no such guarantee. Both policies outperform MSF and First-Fit in all cases.

% We run this configuration with 15 servers, and with jobs grouped into 4 classes, all with $\mu_i =1$. 
% The occurrence probabilities of job classes are $[p_1,p_3,p_5,p_{15}]=[0.5,0.25,0.2,0.05]$.
% By \cref{rem:general-stability}, the stability region is $\lambda < 5$.
%~\Cref{fig:notpof2} shows that in this setting, the weighted mean response times of both Static and Adaptive Quickswap significantly outperform MSF. 
%TODO collapse these, do we need both?

\subsection{Workloads Derived from Google Borg Traces}
\label{sec:sim-borg}

To further evaluate the Adaptive and Static Quickswap policies, we simulate these policies using workloads derived from the Google Borg cluster scheduler traces~\cite{tirmazi2020borg}.
Specifically, we use the methodology of \cite{baiocchi} to extract a workload with the same arrival rates, mean job sizes, and server needs as the Google Borg traces.
\begin{wrapfigure}{r}{0.4\textwidth}
    \vspace{-.1in}
  % \begin{center}
    \includegraphics[width=.48\textwidth]{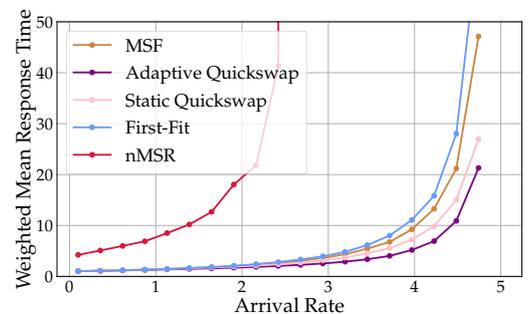}
  % \end{center}
  \vspace{-.2in}
    \caption{Weighted mean response time as a function of job arrival rate in a 4-class MSJ system with $k=15$, $p_1=0.5$, $p_3=0.25$, $p_5=0.2$, and $p_{15}=0.05$.}\label{fig:pof2}
\end{wrapfigure}
\begin{wrapfigure}{r}{0.4\textwidth}
  \vspace{-.2in}
  % \begin{center}
    \includegraphics[width=.48\textwidth]{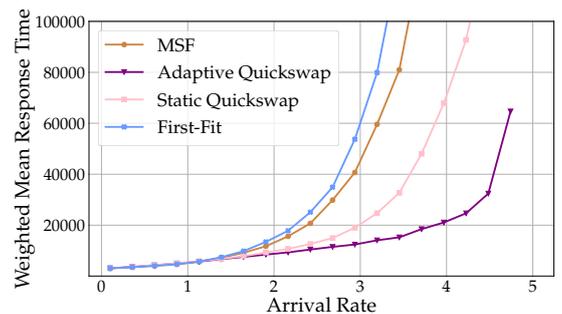}
  % \end{center}
  \vspace{-.2in}
    \caption{Weighted mean response time as a function of job arrival rate for MSJ systems serving a Google Borg workload.  Here, $k=2048$ and the workload is composed of 26 classes based on real-world trace data.}\label{fig:borg}
    \vspace{-.3in}
\end{wrapfigure}
% We rely upon \cite{baiocchi}, which extracts key metrics, including the number of cores used by jobs and the average duration of core usage, from traces of eight distinct machine clusters, or \emph{cells}. 
Our workload consists of 26 job classes from Cell B of the 2019 Borg traces \cite{tirmazi2020borg}.
We set $k$ based on the server need of the heaviest class, so $k=2048$ in our experiments.
The resulting stability region is defined by $\lambda < 4.94$.

% \begin{figure}[t]
% \centering
% \begin{subfigure}[t]{0.48\textwidth}
%     \centering
%     \includegraphics[width=\textwidth]{newplots/cellA.pdf}
%     \caption{Cell A - Weighted mean response time}\label{fig:cellA_w}
% \end{subfigure} 
% \quad
% \begin{subfigure}[t]{0.48\textwidth}
%     \centering
%     \includegraphics[width=\textwidth]{newplots/cellB.pdf}
%     \caption{Cell B - Weighted mean response time}\label{fig:cellB_w}
% \end{subfigure}
% \quad
% \caption{Weighted mean response times versus job arrival rate in MSJ setups using job classes and parameters from Google Borg traces}
% \label{fig:Borg_cells}
% \vspace{-0.2in}
% \end{figure}

\Cref{fig:borg} shows the benefit of using Static and Adaptive Quickswap instead of a competitor policy.
While all policies remain stable, Adaptive and Static Quickswap are once again dominant, improving weighted mean response time by two orders of magnitude when the arrival rate is high.
Note that, due to its poor performance in prior experiments, nMSR is omitted from \Cref{fig:borg}.
These results generally match the trends observed with the synthetic workloads of \Cref{sec:sim-general}, showing a significant benefit obtained with Adaptive Quickswap.
However, even using Static Quickswap provides a 5x reduction in weighted mean response time at high load, compared to the next closest competitor, MSF.
% The fact that in two very different realistic conditions our proposed variations of MSF are able to provide important performance gains is a clear sign of the superiority of our proposed policies with respect to MSF.

%is due to the fact that the biggest job class in the original MSF has a very high response time, and when introducing a weight proportional to the (high) number of requested servers, we obtain a higher overall weighted response time. This could be seen in Fig.~\ref{fig:cellB_w_classes}.

We also compare the unweighted mean response time of MSF, Static and Adaptive Quickswap, and First-Fit in \DIFnomarkup{\ref{app:fairness}}. While MSF achieves good performance at light to medium loads, we show in \DIFnomarkup{\ref{app:fairness}} that this is at the expense of sacrificing certain classes of other jobs by computing the fairness index.
In addition, we show that better fairness metrics and response time performance are possible in \DIFnomarkup{\ref{app:preemption}} by enabling preemption if there is no preemption overhead.
However, preemptions without overheads are unrealistic, therefore we are only focusing on non-preemptive policies.

\section{Conclusion}

This paper describes new, non-preemptive scheduling policies for multiserver jobs.
While the non-preemptive MSF policy has been observed to remain stable at high loads, and while we prove it is throughput-optimal in the one-or-all case, it suffers from high mean response time because it switches service phases too slowly.
We introduce the class of MSFQ policies, and its generalizations, Static Quickswap and Adaptive Quickswap.
These Quickswap policies use a queue length threshold to decide when to switch phases, allowing for the design of policies that switch phases much faster than MSF.
We prove that MSFQ is throughput-optimal in the one-or-all case, analyze the mean response time of MSFQ in the one-or-all case, and demonstrate the benefits of Quickswap policies in simulations based on traces from the Google Borg cluster scheduler.

Before this paper, there were two state-of-the-art choices for non-preemptive multiserver scheduling.
There was MSF, which suffers from slow phase changes, and the class of MSR policies, which use a Markov chain to select schedules and allow phase changes at a faster rate.
The drawback to MSR policies is that their scheduling decisions do not consider queue length information.
As a result, an MSR policy may waste system capacity by reserving servers for jobs that are not in the system.
Hence, before MSFQ, one had to choose between either high resource utilization and slow phase switching, or low resource utilization and fast phase switching.
This paper shows that MSFQ policies get the best of both worlds, using queue length information to switch phases faster without wasting servers.
Although MSFQ is more complex than either of the prior policies, we still provide an accurate analysis of its mean response time in the one-or-all case.

\section*{Acknowledgements}
We thank our shepherd, Dr. Rhonda Righter, and the anonymous reviewers for their insightful feedback
that helped improve our work. This work is supported by a Northwestern IEMS Startup Grant, a UNC Chapel Hill Startup Grant, and National Science Foundation
grants NSF-CCF-2403195 and NSF-IIS-2322974. This work is also supported by TUCAN6-CM (TEC-2024/COM-460), funded by CM, the Region of Madrid, Spain (ORDEN 5696/2024).

% In this paper, we described an important shortcoming of the Most Servers First (MSF) scheduling policy for multiserver-job (MSJ) queuing systems. 
% When load grows, MSF serves sequences of jobs with the same server need for increasingly long intervals of time, leading to poor mean response time.
% % This happens because MSF prioritizes replacing each completed job with another job that uses exactly the same number of cores. As a result, MSF tends to be trapped in serving a same set of job types for very long periods. Changes happen only when a job ends,
% % no job of the same type is waiting, and another job can immediately fit into service.
% % This causes very long waiting times for jobs that do not belong to the types in service.
% In order to remedy this undesirable behavior, we introduced the MSF with Quickswap (MSFQ) policy for the one-or-all setting, and proposed the Static Quickswap and Adaptive Quickswap policies for the general MSJ setting. We proved the throughput-optimality of the MSFQ policy, and derived a high-quality analytical approximation for the mean response time of MSFQ in the one-or-all setting. We empirically demonstrate the accuracy of the analytical model, especially in medium/high load, and demonstrate significant performance improvements of Static Quickswap and Adaptive Quickswap with respect to MSF, especially on Google Borg trace data.

% Analysis of policies that use queue length information, extend to more general resource demand models, such as multidimensional resource requirements.

\clearpage
\bibliographystyle{acm}

\bibliography{biblio}
\appendix

\section{Proof of \cref{thm:pos-stability}}
\label{app:stability}
\stability*
\begin{proof}
    Recall from \cref{sec:MSJ} that we represent the system state as the five-tuple $(n_1, n_k, u_1, u_k, z)$,
        where $z$ is the phase, as defined in \cref{sec:MSFQ}.
        Note that $u_1 \le k$, $u_k \le 1$, and at least one of $u_1$ and $u_k$ must be 0.

        Now, we can define our Lyapunov function $V(\cdot)$:
        \begin{align}
            \label{eq:lyapunov-stability}
            V(n_1, n_k, u_1, u_k, z) :=
              \frac{n_1}{k \mu_1}
            + \frac{n_k}{\mu_k}
            + \mathbbm{1}\{z = 2\} \frac{\epsilon n_k}{2 \lambda_k}
            + \mathbbm{1}\{z \not\in \{1, 2\}\} (c^k - c^{k - u_1}),
        \end{align}
        where $\epsilon = 1 - \frac{\lambda_1}{k \mu_1} -\frac{\lambda_k}{\mu_k}$,
        and where $c = \max\big(2, \frac{\lambda_1}{(\ell - 1) \mu_1 + 1}, \frac{1}{\mu_1}\big)$.

        Intuitively, the two non-indicator terms ensure negative drift in phases 1 and 2,
        where all servers are busy.
        However, they do not suffice for the other phases, when some servers are idle.
        The $\mathbbm{1}\{z = 2\}$ term reduces the negative drift in phase 2 slightly
        to build up a potential,
        which is used in the other phases by the $\mathbbm{1}\{z \not\in \{1, 2\}\}$ term
        to maintain negative drift in that phase.

        We specifically use the continuous-time Foster-Lyapunov theorem~\cite{tweedie_sufficient_1975}.
        We must demonstrate that this Lyapunov function has three properties,
        two of which are defined with reference to the drift $\E[G \circ V(\cdot)]$,
        where $G$ is the instantaneous generator operator of the system.
        \begin{enumerate}
            \item There exists a finite set of states $B$
            and a positive constant $\delta > 0$ such that for all states outside of $B$,
            $\E[G \circ V(\cdot)] \le -\delta$.
            \label{it:stability-1}
            \item There exists a constant $C$ such that
            for all states,
            $\E[G \circ V(\cdot)] \le C$.
            \label{it:stability-2}
            \item There exists a lower bound $D$ such that for all states, $V(\cdot) \ge D$.
            \label{it:stability-3}
        \end{enumerate}
        Demonstrating Property \ref{it:stability-1} is the primary challenge.
        Property \ref{it:stability-3} holds with $D = 0$.
        Property \ref{it:stability-2} follows from the fact that
        $V$ is linear in the two unbounded inputs $n_1, n_k$, which change by at most 1 upon any transition. 
        Furthermore, from any state, the total transition rate is at most $\lambda + \max\{k \mu_1, \mu_k\}$.
        Hence, there exists an upper bound on drift from any system state.

        It thus remains to prove Property \ref{it:stability-1},
        which we show holds with $\delta = \epsilon/2$
        and the following finite exception set $B$,
        consisting of all states in phase 2 in which $n_1 = k$, so phase 2 has the potential to end, and where $n_k$ is below a threshold, as well as the empty state:
        \begin{align*}
            B = \big\{(n_1, n_k, u_1, u_k, z) \mid \big( n_1 = k \,\&\, n_k \le \frac{2 \lambda_k (c^k - 1)}{\epsilon} \,\&\, z = 2 \big) \text{ or } \left(n_1 = 0 \,\&\, n_k = 0\right) \big\}.
        \end{align*}

        We specialize our argument based on the current phase of the system $(z = 1, 2, 3, 4)$,
        and whether the system is in a state on the border of switching phases.

        We start with non-border states in phase 1.
        In this case, $u_k = 1, z=1$. As a result, the drift of $V(\cdot)$ is:
        \begin{align*}
            \E[G \circ V(\cdot)] =
            \frac{\lambda_1}{k \mu_1} + \frac{\lambda_k}{\mu_k} - u_k =
            \frac{\lambda_1}{k \mu_1} + \frac{\lambda_k}{\mu_k} - 1 =
            -\epsilon \le -\delta,
        \end{align*}
        where we use the fact that the drift of $n_1$ is $\lambda_1$ and of $n_k$ is $\lambda_k$.

        Next, consider states in phase 1 that are on the border of switching to phase 2.
        In the subsequent phase 2 state, $n_k = 0$, because the MSFQ policy only switches from serving heavy jobs to
        light jobs when there are no heavy jobs remaining.
        As a result, the $\mathbbm{1}\{z = 2\}$ term in
        \eqref{eq:lyapunov-stability} is 0 when beginning phase 2,
        so \cref{it:stability-1} holds throughout phase 1.

        Next, in phase 2, in states which are not on the border, we have $u_1 = k, z= 2$, by the definition of phase 2 in \cref{sec:MSFQ}. As a result, the drift of $V(\cdot)$ is
        \begin{align*}
            \E[G \circ V(\cdot)] =
            \frac{\lambda_1}{k \mu_1} - \frac{u_1}{k} 
            + \frac{\lambda_k}{\mu_k}
            + \frac{\epsilon \lambda_k}{2 \lambda_k} =
            -\epsilon + \frac{\epsilon}{2}
            = -\epsilon/2 = -\delta.
        \end{align*}

        When the system is in phase 2 and is on the border of switching to phase 3, note that when a phase change occurs, the $\mathbbm{1}\{z = 2\}$ term in \eqref{eq:lyapunov-stability} will change from
        $\frac{\epsilon n_k}{2 \lambda_k}$ to 0,
        and the $\mathbbm{1}\{z = 3\}$ term will change from
        $0$ to $c^k - 1$.
        Recall that the exception set $B$ includes all states in phase 2 on the border where $n_k$ is below the threshold: $n_k \le (c^k - 1) 2 \lambda_k/\epsilon.$
        For all phase 2 border states outside this set, the change in indicators from phase 2 to phase 3
        results in a negative change in $V(\cdot)$, as desired.
        Thus, \cref{it:stability-1} holds throughout phase 2.

        Finally, in the remaining phases 3 and 4, we use different arguments depending on the value of $u_1$.
        We start with states that do not transition directly to phase 1.

        First, in the case where $u_1 = k$, the servers are fully occupied by light jobs.
        In this case, if $n_1 > k$,
        the $\mathbbm{1}\{z = 2\}$ term in \eqref{eq:lyapunov-stability}
        does not change on the next arrival or completion,
        because $u_1 = k$ will remain true after the next event. Thus,
        \begin{align*}
            \E[G \circ V(\cdot)] =
            \frac{\lambda_1}{k \mu_1} - \frac{u_1}{k} 
            + \frac{\lambda_k}{\mu_k} = - \epsilon \le -\delta.
        \end{align*}
        If $u_1 = k$ and $n_1 = k$,
        the indicator term has a negative instantaneous drift:
        If a job completes, $n_1 = k-1$, so the indicator term becomes $c^k - 1$,
        while under any other event, the indicator term remains $c^k$.
        Thus,
        \begin{align*}
            \E[G \circ V(\cdot)] =
            \frac{\lambda_1}{k \mu_1} - \frac{u_1}{k} 
            + \frac{\lambda_k}{\mu_k}
            - k \mu_1 = - \epsilon - k \mu_1 \le -\delta.
        \end{align*}

        In the case where $\ell < u_1 < k$,
        some servers are idle, and the system is in phase 3,
        so light jobs continue to enter service.
        In this case, we upper bound the drift of $V$ as follows:
        \begin{align}
            \nonumber
            \E[G \circ V(\cdot)] &=
            \frac{\lambda_1}{k \mu_1}
            - \frac{u_1}{k}
            + \frac{\lambda_k}{\mu_k}
            + \lambda_1 (c^{k - u_1} - c^{k - (u_1 + 1)})
            + u_1 \mu_1 (c^{k - u_1} - c^{k - (u_1 - 1)}) \\
            &=
            \frac{\lambda_1}{k \mu_1}
            - \frac{u_1}{k}
            + \frac{\lambda_k}{\mu_k}
            + c^{k - u_1}(1-c)(\lambda_1 c^{-1}
            - u_1 \mu_1)
            \label{eq:phase-3-nearly}
            \le
            1
            + c^{k - u_1}(1-c)(\lambda_1 c^{-1}
            - (\ell - 1) \mu_1)\,.
        \end{align}

        Recalling that $c = \max\big(2, \frac{\lambda_1}{(\ell - 1) \mu_1 + 1}\big)$, we substitute into \eqref{eq:phase-3-nearly}.
        As a result,
        \begin{align*}
            \E[G \circ V(\cdot)] \le
            1
            + c^{k - u_1}(1-c)(\lambda_1 c^{-1}
            - (\ell - 1) \mu_1) \le 1 + 2^1 (-1)(1) = -1 \le -\epsilon \le -\delta\,.
        \end{align*}

        In the case where $0 < u_1 \le \ell$: the system is in
        phase 4, and light jobs are blocked from entering service.
        In this case, arriving jobs do not increase $u_1$,
        so the drift is much simpler:
        \begin{align*}
            \E[G \circ V(\cdot)] &=
            \frac{\lambda_1}{k \mu_1}
            - \frac{u_1}{k}
            + \frac{\lambda_k}{\mu_k}
            + u_1 \mu_1 (c^{k - u_1} - c^{k - (u_1 - 1)}) \\
            &\le 1 + u_1 \mu_1 (c^{k - u_1} - c^{k - (u_1 - 1)}) 
            \le 1 + \mu_1 (c - c^2)
            = 1 + \mu_1 c (1-c)\,.
        \end{align*}
        Recalling that $c \ge 2$ and that $c \ge \frac{1}{\mu}$,
        we find that $\E[G \circ V] \le -1$, which completes this step.

        The remaining case within phases 3 and 4 is the case where $u_1 = 0$.
        In this case, the system must be empty,
        because we would otherwise switch to phase 1.
        Recall that this case is in the finite set $B$ of exceptional high-drift states,
        so it does not affect Property \ref{it:stability-1}.

        Having handled states in phases 3 and 4 that do not transition directly to phase 1, we now verify the drift condition for states where the system may transition from phase 3 or 4 to phase 1.
        When entering phase 1, $u_k = 1$, so $u_1 = 0$:
        All light jobs in service have been completed.
        As a result, the $\mathbbm{1}\{z \not\in \{1, 2\}\}$ term in \eqref{eq:lyapunov-stability}
        is 0, as desired.

        With all states verified, Property \ref{it:stability-1} holds,
        so the Foster-Lyapunov theorem demonstrates stability.
\end{proof}

\section{Proofs of \cref{lem:h3} and \cref{lem:h4}}
\label{app:h3-h4}
\tthreeperiod*

\begin{proof}
% Let $\phase_{3, j}$ be the random variable that denotes the period in phase 3 that starts when the system has $i$ light jobs and ends when the system first has $i-1$ light jobs.
To characterize $\phase_{3, j}$, we condition on the first event that happens during the $\phase_{3, j}$ period.
When the first event is an arrival, the remainder of $\phase_{3, j}$ consists of a $\phase_{3, j+1}$ period, followed by another independent $\phase_{3, j}$ period.
When the first event is a completion, the period ends.
In particular, we have 
\begin{align*}
    \phase_{3, j}&=\begin{cases}
        Exp(j\mu_1+\lambda_1)+\phase_{3, j+1}+\phase_{3, j} & \textrm{next event is an arrival}\\
        Exp(j\mu_1+\lambda_1) & \textrm{next event is a departure}
    \end{cases}.
\end{align*}
Note that $\phase_{3, k}\sim B^L_{S_1}$ as the time going from having $k$ to $k-1$ light jobs is the busy period because the completion rate of light jobs here is $k\mu_1$.
Then, the phase duration of phase 3 is the sum of these periods. Formally,
$\phase_{3}=\sum\nolimits_{j=\ell+1}^{k-1}\phase_{3, j}.$
Then, by standard transform techniques, \cref{lem:h3} holds.
\end{proof}

\tfourperiod*

\begin{proof}
    In phase 4, no further light job arrivals are allowed into service, and there are $\ell$ light jobs at the beginning of phase 4.
    Therefore, we can write $\phase_4$ as a sum of i.i.d. exponential distributions: $\phase_4=\sum_{j=1}^\ell Exp(j\mu_1).$
    Therefore, by standard transform techniques,
    $
        \wt{H_4}(s)
        =\prod_{j=1}^\ell \wt{Exp(j\mu_1)}(s)=\prod_{j=1}^\ell \frac{j\mu_1}{j\mu_1+s}.
    $
\end{proof}

\section{Fairness Evaluation}
\label{app:fairness}

To compare the fairness of scheduling policies, we compute Jain's Fairness index \cite{jain1998quantitativemeasurefairnessdiscrimination} as
\begin{equation}
J( \E[T^{(1)}] , \E[T^{(2)}], \cdots , \E[T^{(k)}]) 
=
\frac{\left( \sum_{j=1}^k \E[T^{(j)}] \right)^2}{k \sum_{j=1}^k (\E[T^{(j)}])^2}.
\end{equation}
The value of the fairness index is between $\frac{1}{k}$ and 1.
A higher value of the fairness index indicates that the scheduling policy is fairer.

We examine the fairness of various scheduling policies in \Cref{fig:fairness}.
Although MSF and First-Fit achieve low \emph{unweighted} mean response time (\Cref{fig:cellB_unweighted}), the heavy jobs experience orders of magnitude larger waiting times than the light jobs (\Cref{fig:cellB_sandl}) under these policies.
Under Adaptive and Static Quickswap, on the other hand, the mean response times of light and heavy jobs are comparable.
As a result, Adaptive and Static Quickswap achieve higher fairness indices compared to MSF and First-Fit (\Cref{fig:cellB_fairness}).
Because unweighted mean response time can hide these important imbalances between the job classes, our evaluation uses \emph{weighted mean response time} --- a metric that balances overall mean response time and fairness.

% We see in \Cref{fig:cellB_tot} that the mean response time of MSF is smaller than Static Quickswap at all loads. However, by computing the fairness index of two policies, we can see that Static Quickswap is more fair to all job classes. In other words, it is not optimizing the mean response time at the expense of suffering the performance of a specific job class.

% These results indicate that the unweighted mean response time is not the only meaningful indicator of performance in settings with heterogeneous workloads.  
% By switching phases regularly, Adaptive and Static Quickswap achieve much better fairness than MSF and First-Fit.

\begin{figure}[H]
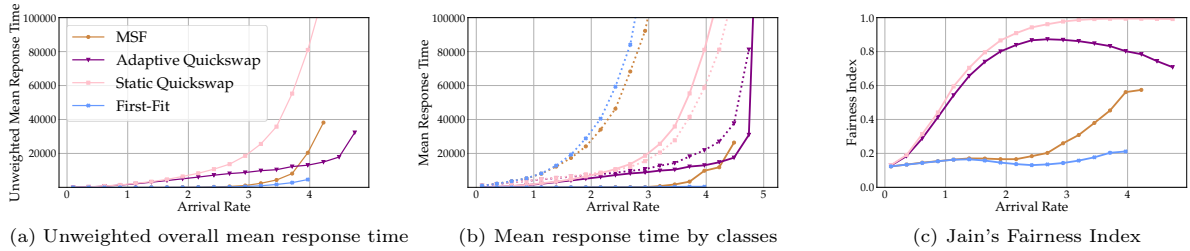

\centering
\begin{subfigure}[t]{0.32\textwidth}
   \centering
   \includegraphics[width=\textwidth]{newplots/cellB_unweighted.pdf}
   \caption{Unweighted overall mean response time}\label{fig:cellB_unweighted}
\end{subfigure}
\begin{subfigure}[t]{0.32\textwidth}
   \centering
   \includegraphics[width=\textwidth]{newplots/cellB_separate.pdf}
   \caption{Mean response time by classes}\label{fig:cellB_sandl}
\end{subfigure}
\begin{subfigure}[t]{0.32\textwidth}
   \centering
   \includegraphics[width=\textwidth]{newplots/cellB_fairness.pdf}
   \caption{Jain's Fairness Index}\label{fig:cellB_fairness}
\end{subfigure}
\caption{The response time performance and fairness index as a function of the overall arrival rate for MSJ systems serving a Google Borg workload. The middle plot, \Cref{fig:cellB_sandl}, shows the mean response time by class, where the dotted lines denote the mean response time of the heaviest jobs and the solid lines denote the mean response time of the lightest jobs. The other plots, left and right, show combined mean response time and fairness metrics across classes in a single, solid line.}
\label{fig:fairness}
\end{figure}

\section{Comparison with Preemptive Policies}
\label{app:preemption}
Although preemption is either infeasible or carries significant overhead for many datacenter workloads, we compare Adaptive and Static Quickswap to a preemptive policy for the sake of completeness.
Here, we assume that the preemptive policy, ServerFilling, can preempt jobs with no overhead or setup cost.
We see that ServerFilling can use preemption to greatly outperform any non-preemptive policy.

\begin{figure}[H]
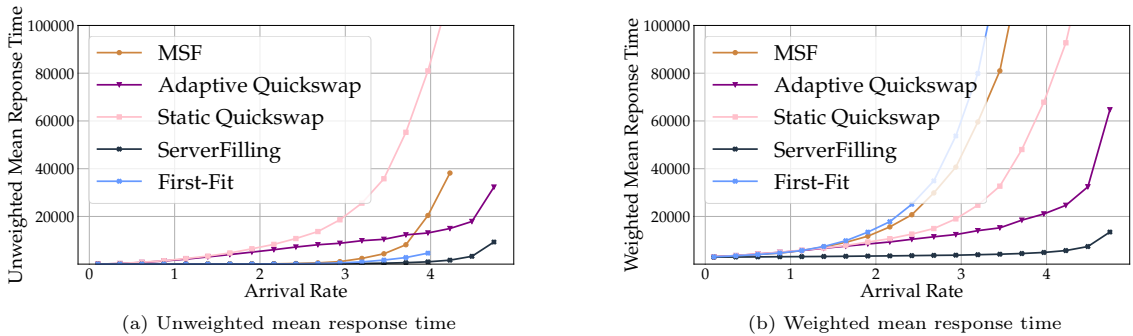

\centering
\begin{subfigure}[t]{0.45\textwidth}
   \centering
   \includegraphics[width=\textwidth]{newplots/sf_cellB_unweighted.pdf}
   \caption{Unweighted mean response time}\label{fig:preemptive-non-weighted}
\end{subfigure}\quad
\quad
\begin{subfigure}[t]{0.45\textwidth}
   \centering
   \includegraphics[width=\textwidth]{newplots/sf_cellB.pdf}
   \caption{Weighted mean response time}\label{fig:preemptive-weighted}
\end{subfigure}\quad
\caption{The overall mean response time of the system as a function of the overall arrival rate for MSJ systems serving a Google Borg workload. Points of higher arrival rates for the First-Fit policy are hidden because the experiments did not converge. }
\label{fig:Borg_cells_2048_rawr}
\end{figure}

Figure \ref{fig:Borg_cells_2048_rawr} shows that the preemptive ServerFilling policy \cite{wcfs} greatly outperforms all non-preemptive policies with respect to both unweighted and weighted mean response time. ServerFilling uses preemptions to guarantee full resource utilization whenever there are more than $k$  jobs in the system. The non-preemptive scheduling policies, on the other hand, may waste significant service capacity even when there are many jobs in the queue.
This waste leads to higher mean response times for all non-preemptive policies, including our Quickswap policies.

\end{document}